\documentclass[onecolumn,12pt,draftclsnofoot]{IEEEtran}
\usepackage{graphicx}
\usepackage{cite}
\usepackage{float}
\usepackage{amsmath}
\usepackage{amsthm}
\usepackage{mathtools}
\usepackage{amssymb}
\usepackage{commath}
\usepackage{subfig}
\usepackage{bbm}
\usepackage{epstopdf}
\usepackage[justification=centering]{caption}
\usepackage{color}
\usepackage{setspace}
\newtheorem{thm}{Theorem}[section]

\newtheorem{lem}[thm]{Lemma}
\newtheorem{definition}{Definition}
\newtheorem{remark}{Remark}

\IEEEoverridecommandlockouts

\ifCLASSINFOpdf
\else
\fi

\begin{document}
\title{Stable Throughput Region of the Two-User Broadcast Channel}
\author{\IEEEauthorblockN{Nikolaos Pappas~\IEEEmembership{Member,~IEEE}, Marios Kountouris~\IEEEmembership{Senior~Member,~IEEE}, Anthony Ephremides~\IEEEmembership{Life~Fellow,~IEEE}, Vangelis Angelakis~\IEEEmembership{Member,~IEEE}}
		
	\thanks{N. Pappas and V. Angelakis are with the Department of Science	and Technology, Link{\"o}ping University, SE-60174 Norrk{\"o}ping, Sweden. Email: \{nikolaos.pappas, vangelis.angelakis\}@liu.se.} 
	\thanks{M. Kountouris is with the Mathematical and Algorithmic Sciences Lab, France Research Center, Huawei Technologies Co. Ltd.
Boulogne-Billancourt, 92100, France. Email: marios.kountouris@huawei.com.}
	\thanks{A. Ephremides is with the Department of Electrical and Computer Engineering and Institute for Systems Research, University of Maryland, College Park, MD 20742. E-mail: etony@umd.edu.}
	\thanks{This work has been partially supported by the People Programme (Marie Curie Actions) of the European Union's Seventh Framework Programme FP7/2007-2013/ under REA grant agreement no.[612361] -- SOrBet.}
	\thanks{This work was presented in part in the IEEE International Conference on Communications (ICC) 2016 \cite{b:pappas-iccBC-2016}.}
}

\maketitle

\begin{abstract}
In this paper we consider the two-user broadcast channel and we characterize its stable throughout region. We start the analysis by providing the stability region for the general case without any specific considerations on transmission and reception mechanisms. We also provide conditions for the stable throughput region to be convex. Subsequently, we consider the case where the transmitter uses superposition coding and we consider two special cases for the receivers. The first one is when both receivers treat interference as noise. The second is when the user with a better channel uses successive decoding and the other receiver treats interference as noise.
\end{abstract}

\newpage

\section{Introduction}

In wireless communications simultaneous transmission to multiple receivers is desired as for example in the case of base stations in cellular systems. The broadcast channel, which was first introduced in \cite{b:CoverBC72}, models the simultaneous communication of information (different messages) from one source to multiple destinations. One approach to establish the communication is to set up orthogonal channels in terms of time/frequency etc. in order to serve each user separately. Although this approach eliminates the interference among concurrent transmissions, it is not in general optimal in terms of achievable rates \cite{b:Tse}. 

Alternative methods include canceling the interference in a variety of ways. For example superposition coding (SC) is a way to remove the orthogonality constraint in a transmission by a base station to a number of receivers and it is known that it may achieve the capacity in the case of a Gaussian broadcast channel \cite{b:CoverBroadcast75} and \cite{b:Cover}.

In this paper we consider an important but often overlooked aspect of operation and performance of the broadcast channel, namely, the stable throughput region
which is related to capacity but concerns the case of random transmission demand patterns \cite{b:EphremidesHajekUnion}.

The exact characterization of the stability region of networks with bursty sources (in contrast to saturated sources, for which the notion of stability does not exist) is known to be a difficult problem due to the interaction among the queues (e.g. each node/queue transmits and, thus, interferes with the others when its queue is non-empty). Thus, we limit our study in this paper to the case of two users so as to capture the essence of the effect of such interaction.
We obtain the stability region for the general case of a broadcast channel without any specific considerations on the details of the transmission and reception structures. We provide conditions for the convexity of the region; this is useful since it implies that simultaneous transmission can outperform time-sharing. We also consider the maximum aggregate stable throughput (known also as sum-throughput) for the general case. Subsequently, we consider the case where the transmitter applies superposition coding and we consider two special cases for the receivers. The first one is when both receivers treat interference as noise. The second is when the user with a better channel uses successive decoding and the other receiver treats interference as noise. For both cases we derive the success probabilities that are needed in order to apply the results obtained for the general case. Two simple transmit power allocation schemes are also considered: i) the assigned power level to each user remains fixed, and ii) the transmit power is adapted to the state of the queue. All the analytical results are also evaluated numerically. 

In Section \ref{sec:model} we describe the system model, in Section \ref{sec:stability_general} we calculate the stability region and obtain the convexity conditions. In Section \ref{sec:Aggregate} the maximum aggregate stable throughput is obtained for the general case. In Sections \ref{sec:stability_IAN} and \ref{sec:stability_SC} we consider the case where the transmitter uses superposition coding and either both receivers treat interference as noise or the stronger receiver applies successive decoding and the other treats interference as noise. In Section \ref{sec:NumResults} we provide numerical evaluation of the analytical results presented in the aforementioned sections. 

The broadcast channel was introduced in \cite{b:CoverBC72}. The work in \cite{b:CoverBC72}, \cite{b:CoverBroadcast75}, \cite{b:Marton}, \cite{b:Jafarian} provide a characterization of the information-theoretic capacity region of the broadcast channel. Caire and Shamai in \cite{b:CaireShamaiToIT2003} investigated the achievable throughput of a multi-antenna Gaussian broadcast channel. Fayolle et al. \cite{b:FayolleBroadcast} provided a theoretical treatment of some basic problems related to the stability investigation of the broadcast channel. 
In \cite{b:ZhouWunderAllerton2008}, scheduling policies in a broadcast system were considered and general conditions covering a class of throughput optimal scheduling policies were obtained. In \cite{b:JafarBC_stability}, the authors characterized the stability regions of two-user Gaussian fading multiple access and broadcast channels with centralized scheduling under the assumption of backlogged users. In \cite{b:EphremidesBroadcast13}, the capacity region of the two-user broadcast erasure channel was characterized and algorithms based on linear network coding and its stability region was also obtained.

\section{System Model} \label{sec:model}
We consider a two-user broadcast channel, as shown in Fig.\ref{fig:system_model}, in which one transmitter $S$ has two different queues with packets intended for two different receivers. The $i$-th queue ($i=1,2$) contains the packets (messages) that are destined to receiver $D_i$. Time is assumed to be slotted, the packet arrival processes at the first and the second queue are assumed to be independent and stationary with mean rates $\lambda_1$ and $\lambda_2$ in packets per slot, respectively. Both queues have infinite capacity to store incoming packets and $Q_i$ denotes the size, measured in number of packets. The source transmits packets in a time slot if at least one of its queue is not empty. The transmission of one packet requires one time slot and we assume that acknowledgements (ACKs) are received instantaneously and error-free (this is a simplifying, standard assumption in studies of this kind).

\begin{figure}[]
\centering
\includegraphics[scale=1.2]{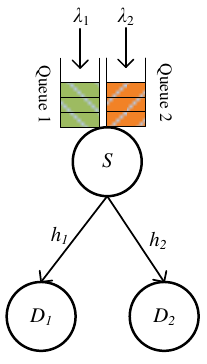}
\caption{The two-user broadcast channel with bursty arrivals.}
\centering
\label{fig:system_model}
\end{figure}

If only one queue at the source is non-empty during a time slot, then the transmitter sends information to the corresponding receiver only. When both queues have packets, the source transmits a packet that contains both messages, whereas whenever both queues at the source are empty, the transmitter remains silent.

Let $\mathcal{D}_{i/\mathcal{T}}$ denote the event that $D_i$ is able to decode successfully the packet transmitted from the $i$-th queue of the transmitter given a set of non-empty queues denoted by $\mathcal{T}$, e.g. $\mathcal{D}_{1/1,2}$ denotes the event that receiver $D_1$ can decode the packet from the first queue when both queues are not empty ($\mathcal{T} = \{1,2\}$). It is reasonably assumed that no matter what the detection mechanism is we have $\mathrm{Pr}\left(\mathcal{D}_{1/1,2} \right) \leq \mathrm{Pr}\left(\mathcal{D}_{1/1} \right)$.

The average packet service rate for the first queue is 
\begin{eqnarray} \label{eq:mu1}
\mu_1 = \mathrm{Pr}\left(Q_2 > 0 \right) \mathrm{Pr}\left(\mathcal{D}_{1/1,2} \right) + \mathrm{Pr}\left(Q_2 = 0 \right) \mathrm{Pr}\left(\mathcal{D}_{1/1} \right).
\end{eqnarray}
Respectively, the average service rate for the second queue is
\begin{eqnarray} \label{eq:mu2}
\mu_2 = \mathrm{Pr}\left(Q_1 > 0 \right) \mathrm{Pr}\left(\mathcal{D}_{2/1,2} \right) + \mathrm{Pr}\left(Q_1 = 0 \right) \mathrm{Pr}\left(\mathcal{D}_{2/2} \right).
\end{eqnarray}
If a packet from the $i$-th queue fails to reach $D_i$, it remains in the queue and is retransmitted in the next time slot. 

If only the $i$-th queue is non-empty the received signal $y_i$ of user $D_i$ is given by 
\begin{eqnarray} \label{eq:signal}
y_i = h_i x_i + n_i,
\end{eqnarray}
where $n_i$ is the additive white Gaussian noise with zero mean and unit variance. The channel gain from the transmitter to $D_i$ is denoted by $h_i$, and the transmitted signal is $x_i$. A block fading model with Rayleigh fading is considered for the channel, i.e. the fading coefficients $h_i$ remain constant during one time slot, but change independently from one time slot to another. The transmitter allocates power $P_i$ for messages (packets) from the $i$-th queue. 

The event $\mathcal{D}_{i/i}$ is defined as the probability that the received signal-to-noise ratio (SNR) of the $i$-th receiver is above a certain threshold $\gamma_i$, i.e. $\mathcal{D}_{i/i} = \left\lbrace \gamma_i \leq \mathrm{SNR}_i \right\rbrace$. The distance between the transmitter and $D_i$ is denoted by $d_i$. Under the physical model, $\mathrm{SNR}_i \triangleq |h_{i}|^2 d^{-\alpha}_{i} P_i$, where a distance-dependent, power-law path loss function $d_i^{\alpha}$ is considered, with $\alpha$ being the path loss exponent. A transmission is successful if and only if SNR at the intended receiver exceeds a threshold so that the transmitted signal can be decoded with an acceptable bit error probability. The probability that the link between the transmitter and $D_i$ is not in outage when only the $i$-th queue is non-empty is given by (Ch. 5.4 in~\cite{b:Tse})
\begin{eqnarray} \label{eq:SNR}
\mathrm{Pr}\left(\mathcal{D}_{i/i}\right) =\mathrm{Pr} \left\lbrace \mathrm{SNR}_i \geq \gamma_i \right\rbrace = \exp \left(- \frac{\gamma_i d^{\alpha}_{i}}{P_{i}}\right).
\end{eqnarray}

The main result of this paper is oblivious to the details of how successful reception is achieved. It is just based on the values of the success probabilities which may depend on several factors like power, rate, distance, coding and decoding algorithms. Those are considered in the subsequent sections.

When both queues at the source are non-empty, the source transmits the superimposed signal $x=x_1 +x_2$, where $x_i$ is the signal intended for $D_i, i=1,2$. Then, the received signal $y_i$ at $D_i$ is given by $y_i = h_i x + n_i$, and the total transmit power is $P=P_1+P_2$. 
We refer to the two packets used in a single superposition-based transmission as two levels. The packet intended for the weak receiver (i.e. $D_2$) is referred to as the first level. We refer to the other level as the second level. A transmitter using superposition coding splits the available transmission power between the two levels, selects the transmission rate for each of the levels, then encodes and modulates each of the packets separately at the selected rate. 
The modulated symbols are scaled appropriately to match the chosen power split and constraint, and summed to obtain the transmitted signal. More details about implementation of superposition coding at the medium access layer can be found in \cite{b:SC_Scheduling_Haenggi} and \cite{b:SCMobicom2007}. 
We assume that the transmitter and both receivers $D_i$ know perfectly each channel realization $h_i$ (perfect CSI).

\subsection{Stability Criterion}

We use the following definition of queue stability~\cite{Szpankowski:stability}:

\begin{definition}
Denote by $Q_i^t$ the length of queue $i$ at the beginning of time slot $t$. The queue is said to be \emph{stable} if
$\lim_{t \rightarrow \infty} {Pr}[Q_i^t < {x}] = F(x)$ and $\lim_{ {x} \rightarrow \infty} F(x) = 1$.
\end{definition}

Although we do not make explicit use of this definition we use its corollary consequence which is Loynes' theorem~\cite{b:Loynes} that states that if the arrival and service processes of a queue are strictly jointly stationary and the average arrival rate is less than the average service rate, then the queue is stable. If the average arrival rate is greater than the average service rate, then the queue is unstable and the value of $Q_i^t$ approaches infinity almost surely. The stability region of the system is defined as the set of arrival rate vectors $\boldsymbol{\lambda}=(\lambda_1, \lambda_2)$ for which the queues in the system are stable.

\section{The Stability Region -- The General Case} \label{sec:stability_general}

The average service rates of the first and second queue are given by (\ref{eq:mu1}) and (\ref{eq:mu2}), respectively.
Since the average service rate of each queue depends on the queue size of the other queue, the stability region cannot be computed directly. For that, we apply the stochastic dominance technique~\cite{rao:stability}, i.e. we construct hypothetical dominant systems, in which the source transmits dummy packets for the packet queue that is empty, while for the non-empty queue it transmits according to its traffic.

\subsection{First Dominant System: the first queue transmits dummy packets}
In the first dominant system, when the first queue is empty, the source transmits a dummy packet for $D_1$, while the second queue behaves as in the original system. All other operational aspects remain unaltered in the dominant system. Thus, in this dominant system, the first queue never empties, hence the service rate for the second queue is constant and given by $\mu_2 = \mathrm{Pr}\left(\mathcal{D}_{2/1,2} \right)$.

Then, we can obtain stability conditions for the second queue by applying Loynes' criterion~\cite{b:Loynes}. The queue at the second source is stable if and only if $\lambda_2 < \mu_2$, that is $\lambda_2 < \mathrm{Pr}\left(\mathcal{D}_{2/1,2}\right)$.
Then we can obtain the probability that the second queue is empty by applying Little's theorem, i.e.
\begin{eqnarray} \label{eq:Pr2empty_D1}
\mathrm{Pr}\left(Q_2 = 0 \right)  = 1-\frac{\lambda_2}{\mathrm{Pr}\left(\mathcal{D}_{2/1,2}\right)}.
\end{eqnarray}

Inserting (\ref{eq:Pr2empty_D1}) into (\ref{eq:mu1}), we obtain that the service rate for the first queue in the first dominant system is given by
\begin{eqnarray} \label{eq:mu1_D1}
\mu_1 = \mathrm{Pr}\left(\mathcal{D}_{1/1}\right) - \frac{\mathrm{Pr}\left(\mathcal{D}_{1/1}\right) - \mathrm{Pr}\left(\mathcal{D}_{1/1,2}\right)}{\mathrm{Pr}\left(\mathcal{D}_{2/1,2}\right)}\lambda_2.
\end{eqnarray}

The first queue is stable if and only if $\lambda_1 < \mu_1$. The stability region $\mathcal{R}_1$ obtained from the first dominant system
is given by

\begin{align} \label{eq:R_1}
\mathcal{R}_1 = \left\lbrace (\lambda_{1},\lambda_{2}): \frac{\lambda_1}{\mathrm{Pr}\left(\mathcal{D}_{1/1} \right)} + \frac{\mathrm{Pr}\left(\mathcal{D}_{1/1} \right) - \mathrm{Pr}\left(\mathcal{D}_{1/1,2} \right)}{\mathrm{Pr}\left(\mathcal{D}_{1/1} \right)\mathrm{Pr}\left(\mathcal{D}_{2/1,2} \right)}\lambda_2 < 1  , \lambda_2 < \mathrm{Pr} \left(\mathcal{D}_{2/1,2} \right)  \right\rbrace.
\end{align}
 
\subsection{Second Dominant System: the second queue transmits dummy packets}

In the second dominant system, when the second queue empties, the source transmits a dummy packet for $D_2$, while the first queue behaves in the same way as in the original system. In this dominant system, the second queue is never empty, so the service rate for the first queue is given by
\begin{eqnarray} \label{eq:mu1_D2}
\mu_1 = \mathrm{Pr}\left(\mathcal{D}_{1/1,2}\right).
\end{eqnarray}

The first queue is stable if and only if $\lambda_1 < \mu_1$. The probability that $Q_1$ is empty is given by
\begin{eqnarray} \label{eq:Pr1empty_D2}
\mathrm{Pr}\left(Q_1 = 0 \right)  = 1-\frac{\lambda_1}{\mathrm{Pr}\left(\mathcal{D}_{1/1,2}\right)}.
\end{eqnarray}

The service rate of the second queue, after substituting (\ref{eq:Pr1empty_D2}) into (\ref{eq:mu2}) is given by
\begin{equation} \label{eq:mu2_D2}
\mu_2 = \mathrm{Pr}\left(\mathcal{D}_{2/2}\right) - \frac{\mathrm{Pr}\left(\mathcal{D}_{2/2}\right) - \mathrm{Pr}\left(\mathcal{D}_{2/1,2}\right)}{\mathrm{Pr}\left(\mathcal{D}_{1/1,2}\right)}\lambda_1.
\end{equation}

The stability region $\mathcal{R}_2$ obtained from the second dominant system is given by
\begin{align} \label{eq:R_2}
\mathcal{R}_2 = \left\lbrace (\lambda_{1},\lambda_{2}): \frac{\lambda_2}{\mathrm{Pr}\left(\mathcal{D}_{2/2} \right)} + \frac{\mathrm{Pr}\left(\mathcal{D}_{2/2} \right) - \mathrm{Pr}\left(\mathcal{D}_{2/1,2} \right)}{\mathrm{Pr}\left(\mathcal{D}_{2/2} \right)\mathrm{Pr}\left(\mathcal{D}_{1/1,2} \right)}\lambda_1 < 1  , \lambda_1 < \mathrm{Pr} \left(\mathcal{D}_{1/1,2} \right)  \right\rbrace.
\end{align}

The stability region of the system is given by $\mathcal{R} = \mathcal{R}_1 \bigcup \mathcal{R}_2$, where 
$\mathcal{R}_1$ and $\mathcal{R}_2$ are given by (\ref{eq:R_1}) and (\ref{eq:R_2}) respectively and is depicted in Fig. \ref{fig:region_general}.

\begin{figure}[]
\centering
\includegraphics[scale=0.6]{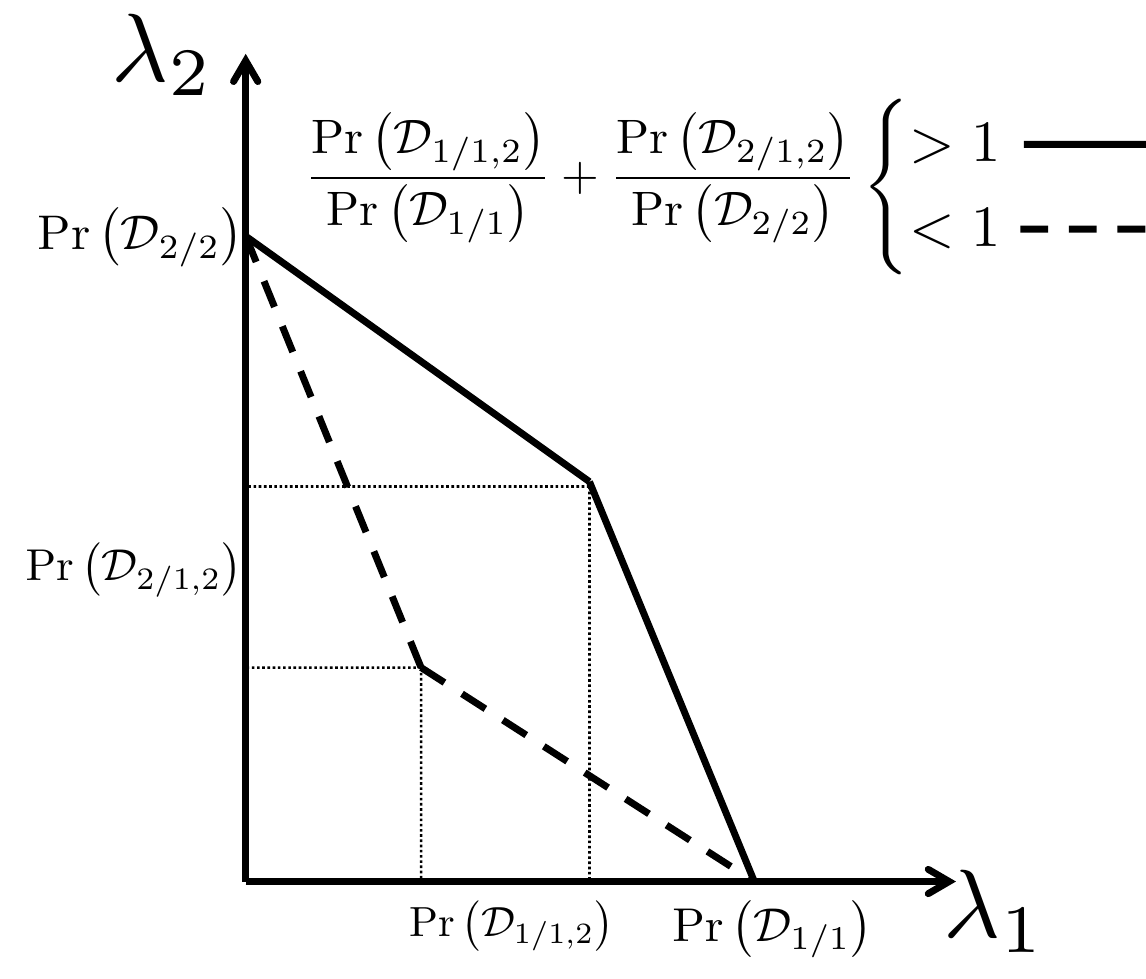}
\caption{The stability region for the two-user broadcast channel in the general case.}
\centering
\label{fig:region_general}
\end{figure}

An important observation made in \cite{rao:stability} is that the stability conditions obtained by the stochastic dominance technique are not only sufficient but also necessary for the stability of the original system. The \emph{indistinguishability} argument~\cite{rao:stability} applies to our problem as well. Based on the construction of the dominant system, it is easy to see that the queue sizes in the dominant system are always greater than those in the original system, provided they are both initialized to the same value and the arrivals are identical in both systems. Therefore, given $\lambda_{2}<\mu_{2}$, if for some $\lambda_{1}$, the queue at $S_1$ is stable in the dominant system, then the corresponding queue in the original system must be stable. Conversely, if for some $\lambda_{1}$ in the dominant system, the queue at node $S_1$ saturates, then it will not transmit dummy packets, and as long as $S_1$ has a packet to transmit, the behavior of the dominant system is identical to that of the original system since dummy packet transmissions are eliminated as we approach the stability boundary. Therefore, the original and the dominant system are indistinguishable at the boundary points. This line of argument has been used in almost all papers on the subject of interacting queues.

Interestingly, the obtained stability region for the two-user broadcast channel in the general case has the same expression as the stability region of the two-user interference channel obtained in \cite{b:pappasitw2013-ic}. Note that the stability region was obtained without any specific assumption on the physical layer processing and the receiver. 

The stability region $\mathcal{R}$ is a 2-dimensional {\emph{convex polyhedron}} when the following condition holds:
\begin{eqnarray}\label{eq:convexity_general}
\frac{\mathrm{Pr}\left(\mathcal{D}_{1/1,2}\right)}{\mathrm{Pr}\left(\mathcal{D}_{1/1}\right)}+
\frac{\mathrm{Pr}\left(\mathcal{D}_{2/1,2}\right)}{\mathrm{Pr}\left(\mathcal{D}_{2/2}\right)} \geq 1.
\end{eqnarray}

When equality holds in (\ref{eq:convexity_general}), the region is a triangle and coincides with the case of time-sharing. Convexity is an important property since it corresponds to the case when parallel concurrent transmissions are preferable to a time-sharing scheme. Additionally, convexity of the stability region implies that if two rate pairs are stable, then any rate pair lying on the line segment joining those two rate pairs is also stable. Convexity of the region also results in higher aggregate stable throughput as discussed in Section \ref{sec:Aggregate}.

\subsection{Closure of the Stability Region} \label{sec:closure}
In the previous subsection, we obtained the stability region in terms of success probabilities under the assumption of fixed powers.
If we take the union of these regions over all possible power allocations to the users we obtain the total stability region (i.e. the envelope of the individual regions). This corresponds to the closure of the stability region and is defined as

\begin{equation} \label{eq:closure_def}
\mathcal{L}\triangleq \left( \bigcup_{ (P_1, P_2) \in [0,P]^2, \\ P_1+P_2=P} \mathcal{L}_1 (P_1,P_2) \right) \bigcup \left( \bigcup_{ (P_1, P_2) \in [0,P]^2, \\ P_1+P_2=P} \mathcal{L}_2 (P_1,P_2) \right)
\end{equation}
where $\mathcal{L}_i (P_1,P_2) \triangleq \mathcal{R}_i$ for $i=1,2$ are obtained in the previous subsections.

\section{Maximum Aggregate Stable Throughput} \label{sec:Aggregate}
In addition to the stability region, another important performance metric is the maximum aggregate stable throughput $T_\mathrm{aggr}$, i.e. the sum of the arrivals rates such that both queues are stable as stated by
\begin{eqnarray*}
\begin{aligned}
& {\text{maximize}}
& & T_\mathrm{aggr} = \lambda_1+\lambda_2 \\
& \text{subject to}
& & (\lambda_{1},\lambda_{2}) \in \mathcal{R}.
\end{aligned}
\end{eqnarray*}

The above is a trivial maximization problem that corresponds to a simple linear program, hence the optimal solution lies at an extreme point (corner point of the stability region) \cite{b:Bertsimas}. In Section \ref{sec:stability_general} the stability region was obtained, whose corner points are $\left(0,\mathrm{Pr}\left(\mathcal{D}_{2/2}\right) \right)$, $\left( \mathrm{Pr}\left(\mathcal{D}_{1/1,2}\right),\mathrm{Pr}\left(\mathcal{D}_{2/1,2}\right) \right)$ and $\left(\mathrm{Pr}\left(\mathcal{D}_{1/1}\right),0 \right)$. Thus, the solution to the aforementioned optimization problem is

\begin{equation} 
T_\mathrm{aggr} = \max \left\{\mathrm{Pr}\left(\mathcal{D}_{2/2}\right), \mathrm{Pr}\left(\mathcal{D}_{1/1,2}\right)+\mathrm{Pr}\left(\mathcal{D}_{2/1,2}\right), \mathrm{Pr}\left(\mathcal{D}_{1/1}\right) \right\}.
\end{equation}

\section{Treating Interference as Noise}  \label{sec:stability_IAN}
In this section, we consider the case where the users treat the interfering signal as noise. When the $i$-th queue is empty at the source, while the $j$-th queue is not, the success probability for the $i$-th user is given by (\ref{eq:SNR}). 
When both queues are non-empty then the transmitted signal at a time slot from the source to the receivers is denoted by $x=x_1+x_2$. The received signal $y_i$ by the user $D_i$ is 
$y_i = h_i x + n_i$. The event $\mathcal{D}_{i/i,j}$ denotes that user $D_i$ is able to decode its intended packet, when both queues are non-empty. This is feasible when the received SINR is above a threshold $\gamma_i$ and is expressed by

\begin{equation} \label{eq:SINR_IAN}
\mathcal{D}_{i/i,j}= \left\lbrace \frac{P_i \abs{h_i}^2 d^{-\alpha}_{i} }{1+P_j\abs{h_i}^2 d^{-\alpha}_{i}} \geq \gamma_i \right\rbrace.
\end{equation}

The following lemma provides the success probability of the second user when both queues are non-empty.
Similarly we obtain the success probability for the first user.
\begin{lem}
The success probability of the second user, $D_2$ when both queues are non-empty is given by
\begin{align} \label{eq:Pr212}
\mathrm{Pr}\left(\mathcal{D}_{2/1,2} \right) = \mathbbm{1}\left\{ P_2 > \gamma_2 P_1 \right\} \exp \left(- \frac{\gamma_2 d_2 ^{\alpha}}{ P_2 -\gamma_2 P_1 } \right),
\end{align}
where $\mathbbm{1} \{\cdot\}$ is the indicator function.
\end{lem}

\begin{proof}
The transmission from the source to $D_2$ when both queues are non-empty is successful when 
\begin{eqnarray*}
\frac{P_2 \abs{h_2}^2 d^{-\alpha}_{2} }{1+P_1\abs{h_2}^2 d^{-\alpha}_{2}} \geq \gamma_2  \iff \gamma_2 \leq \abs{h_2}^2 d^{-\alpha}_{2} (P_2 -\gamma_2 P_1) \\
\iff \gamma_2 \leq \abs{h_2}^2 d^{-\alpha}_{2} (P_2 -\gamma_2 (P-P_2)).
\end{eqnarray*}
Note that $P_1+P_2=P$. Thus, if $P_2 -\gamma_2 (P-P_2) < 0$ then the success probability is zero because the initial inequality is not feasible. 
Thus, if $P_2 > \frac{\gamma_2}{1+\gamma_2}P$ then
\begin{eqnarray*}
\abs{h_2}^2 \geq \frac{\gamma_2}{d^{-\alpha}_{2} (P_2 -\gamma_2 (P-P_2))}.
\end{eqnarray*}
Assuming Rayleigh block fading, the channel power is exponentially distributed, i.e. $\abs{h_2}^2 \sim \exp(1)$, and the success probability can be expressed as
\begin{equation}
\begin{aligned}
\mathrm{Pr} \left(\mathcal{D}_{2/1,2} \right)=\mathrm{Pr} \left[\abs{h_2}^2 \geq \frac{\gamma_2}{d^{-\alpha}_{2} (P_2 -\gamma_2 (P-P_2))}  \right] = \\ = \int_0^{\infty} \! \mathrm{Pr} \left[x \geq \frac{\gamma_2}{d^{-\alpha}_{2} (P_2 -\gamma_2 (P-P_2))}\right] f_{\abs{h_2}^2}(x)  \,  \mathrm{d}x.
\end{aligned}
\end{equation}

Thus,
\begin{equation}
\begin{aligned}
\mathrm{Pr} \left(\mathcal{D}_{2/1,2} \right)= \int_0^{\infty} \! \left[ 1 - F_{\abs{h_2}^2} \left(\frac{\gamma_2}{d^{-\alpha}_{2} (P_2 -\gamma_2 (P-P_2))}\right) \right] f_{\abs{h_2}^2}(x)  \, \mathrm{d}x,
\end{aligned}
\end{equation}
with $f_{\abs{h_2}^2}(x)= \exp(-x)$ and $F_{\abs{h_2}^2} (x) = 1-\exp(-x)$.

To summarize, the success probability for the second user when both queues at the source are non-empty is given by (\ref{eq:Pr212}), where $\mathbbm{1} \{\cdot\}$ is the indicator function.
\end{proof}

Thus if $P_2 > \gamma_2 P_1$ and $P_1 > \gamma_1 P_2$ then $\mathrm{Pr}\left(\mathcal{D}_{1/1,2} \right)=\exp \left(- \frac{\gamma_1 d_1 ^{\alpha}}{ P_1 -\gamma_1 P_2 } \right)$ and $\mathrm{Pr}\left(\mathcal{D}_{2/1,2} \right)=\exp \left(- \frac{\gamma_2 d_2 ^{\alpha}}{ P_2 -\gamma_2 P_1 } \right)$ and after combining (\ref{eq:R_1}) and (\ref{eq:R_2}) we obtain the stability region $\mathcal{R} = \mathcal{R}_1 \bigcup \mathcal{R}_2$.
 
\begin{remark}
Note that $\mathrm{Pr}\left(\mathcal{D}_{1/1,2} \right) \neq 0$ and $\mathrm{Pr}\left(\mathcal{D}_{2/1,2} \right)\neq 0$ if and only
if $\gamma_1 \gamma_2 \leq 1$.
\end{remark}
The previous condition can be obtained by the conditions $P_2 > \gamma_2 P_1$, $P_1 > \gamma_1 P_2$ and $P_1+P_2=P$.

\section{Successive Decoding}  \label{sec:stability_SC}
In this section we consider a decoding strategy where the receiver with the better channel applies successive decoding and 
the other one treats interference as noise. More specifically, we assume that the channel from the transmitter to $D_1$ is better than that to $D_2$.

At the receiver side, $D_2$ treats the message of $D_1$ as noise and decodes its data from $y_2$. 
Receiver $D_2$ is able to decode its intended packet if and only if the received signal-to-interference-plus-noise ratio (SINR) is greater than $\gamma_2$.
The probability that the communication between the transmitter and $D_2$ is successful when both queues are non-empty is given by (\ref{eq:Pr212}). 

Receiver $D_1$, which has a better channel, performs successive decoding, i.e. it decodes first the message of $D_2$, then it subtracts it from the received signal, and proceeds to decode its message with a single-user decoder. 

The following lemma provides the success probability of the first user when both queues are non-empty.

\begin{lem}
The success probability of the first user, $D_1$ when both queues are non-empty is given by
\begin{align} \label{eq:Pr112_SC}
\mathrm{Pr}\left(\mathcal{D}_{1/1,2} \right) = \mathbbm{1}\left\{ \gamma_2 P_1 < P_2 \leq  P_1 \frac{\gamma_2 (1+\gamma_1)}{\gamma_1} \right\} \exp \left(- \frac{\gamma_2 d_1 ^{\alpha}}{ P_2 -\gamma_2 P_1 } \right) + \mathbbm{1}\left\{ P_2 > P_1 \frac{\gamma_2 (1+\gamma_1)}{\gamma_1} \right\} \exp \left(- \frac{\gamma_1 d^{\alpha}_{1}}{P_{1}}\right)
\end{align}
where $\mathbbm{1} \{\cdot\}$ is the indicator function.
\end{lem}

\begin{proof}
The successive decoding is feasible at the first receiver if 

\begin{eqnarray}
\left\lbrace \frac{P_2 \abs{h_1}^2 d^{-\alpha}_1} {1+P_1\abs{h_1}^2 d^{-\alpha}_{1}} \geq \gamma_2 ,\text{ } P_1  \abs{h_1}^2 d^{-\alpha}_{1} \geq \gamma_1 \right\rbrace.
\end{eqnarray}

From the first condition we have that
\begin{eqnarray*}
\frac{P_2 \abs{h_1}^2 d^{-\alpha}_{1} }{1+P_1\abs{h_1}^2 d^{-\alpha}_{1}} \geq \gamma_2  \iff \gamma_2 \leq \abs{h_1}^2 d^{-\alpha}_{1} (P_2 -\gamma_2 P_1). 
\end{eqnarray*}

If $P_2 \leq \gamma_2 P_1$, the previous condition is not feasible, thus $D_1$ is not able to decode the information in this case.
If $P_2 > \gamma_2 P_1$ then successive decoding is feasible at the first receiver if 

\begin{eqnarray}
\left\lbrace \abs{h_1}^2 \geq \frac{\gamma_2  d^{\alpha}_{1}}{P_2 -\gamma_2 P_1}  ,\text{ } \abs{h_1}^2  \geq \frac{\gamma_1 d^{\alpha}_{1}}{P_1} \right\rbrace \iff \abs{h_1}^2 \geq \max \left\lbrace\frac{\gamma_2  d^{\alpha}_{1}}{P_2 -\gamma_2 P_1}, \frac{\gamma_1 d^{\alpha}_{1}}{P_1}\right\rbrace.
\end{eqnarray}

The success probability can be expressed as
\begin{equation}
\begin{aligned}
\mathrm{Pr} \left(\mathcal{D}_{1/1,2} \right)=\mathrm{Pr} \left[\abs{h_1}^2 \geq \max \left\lbrace\frac{\gamma_2  d^{\alpha}_{1}}{P_2 -\gamma_2 P_1}, \frac{\gamma_1 d^{\alpha}_{1}}{P_1}\right\rbrace  \right] = \\ = \int_0^{\infty} \! \mathrm{Pr} \left[x \geq \max \left\lbrace\frac{\gamma_2  d^{\alpha}_{1}}{P_2 -\gamma_2 P_1}, \frac{\gamma_1 d^{\alpha}_{1}}{P_1}\right\rbrace\right] f_{\abs{h_1}^2}(x)  \,  \mathrm{d}x.
\end{aligned}
\end{equation}

Thus,
\begin{equation}
\begin{aligned}
\mathrm{Pr} \left(\mathcal{D}_{1/1,2} \right)= \int_0^{\infty} \! \left[ 1 - F_{\abs{h_1}^2} \left(\max \left\lbrace\frac{\gamma_2  d^{\alpha}_{1}}{P_2 -\gamma_2 P_1}, \frac{\gamma_1 d^{\alpha}_{1}}{P_1}\right\rbrace\right) \right] f_{\abs{h_1}^2}(x)  \, \mathrm{d}x.
\end{aligned}
\end{equation}
Note that $f_{\abs{h_1}^2}(x)= \exp(-x)$ and $F_{\abs{h_1}^2} (x) = 1-\exp(-x)$.

From the above we obtain that the success probability is given by 
\begin{equation}
\mathrm{Pr} \left(\mathcal{D}_{1/1,2} \right)= \begin{cases} \exp \left(- \frac{\gamma_2 d_1 ^{\alpha}}{ P_2 -\gamma_2 P_1 } \right) , & \mbox{if } \gamma_2 P_1 < P_2 \leq  P_1 \frac{\gamma_2 (1+\gamma_1)}{\gamma_1}  \\ \exp \left(- \frac{\gamma_1 d^{\alpha}_{1}}{P_{1}}\right), & \mbox{if } P_2 > P_1 \frac{\gamma_2 (1+\gamma_1)}{\gamma_1}. \end{cases}
\end{equation}
The success probability is summarized in (\ref{eq:Pr112_SC}).
\end{proof}

\begin{remark}
From (\ref{eq:Pr112_SC}), we observe that if the following condition is satisfied
\begin{eqnarray} \label{eq:condition_1stRx}
P_2 > P_1 \frac{\gamma_2 (1+\gamma_1)}{\gamma_1},
\end{eqnarray}
then the link success probability for $D_1$ when both queues are non-empty is given by
\begin{eqnarray}
\mathrm{Pr}\left(\mathcal{D}_{1/1,2}\right) =\mathrm{Pr}\left(\mathcal{D}_{1/1}\right) = \exp \left(- \frac{\gamma_1 d^{\alpha}_{1}}{P_{1}}\right).
\end{eqnarray}
This means that the transmission for the first user is not affected by the transmission to the second one.
\end{remark}

In the remainder, we consider two simple schemes regarding the assigned transmission power for each receiver's queues. The first scheme corresponds to the case where we fix the transmission power $P_i$ for the $i$-th receiver, such that $P_1+P_2=P$, i.e. even if the second queue is empty, the assigned power to the first will remain $P_1$. The second scheme is rather natural when a user's queue is empty. Thus, the transmitter adapts the power according to the queue state of each receiver, i.e. if one queue is empty, then all power $P$ is allocated to the other user's signal\footnote{This power allocation scheme can be also applied when interference is treated as noise at both receivers (see Section \ref{sec:stability_IAN}), whose performance is evaluated in Section \ref{sec:NumResults}.}.

\subsection{Fixed Power Scheme}
We assume here that the transmitter assigns fixed power $P_1$ (resp. $P_2$) at the $D_1$ (resp. $D_2$) on every time slot.
\subsubsection{The case where $P_2 > P_1 \frac{\gamma_2 (1+\gamma_1)}{\gamma_1}$}
The service rate seen by the first queue is given by (\ref{eq:mu1}).
Since constant transmit power $P_1$ is used and $D_1$ has better channel than $D_2$, from (\ref{eq:Pr112_SC}) we have that $\mathrm{Pr}\left(\mathcal{D}_{1/1,2} \right) = \mathrm{Pr} \left(\mathcal{D}_{1/1} \right)$. Thus, we have
\begin{eqnarray}
\mu_1 = \mathrm{Pr} \left(\mathcal{D}_{1/1} \right).
\end{eqnarray}

From Loynes' criterion for stability~\cite{b:Loynes}, the first queue is stable if and only if $\lambda_1 < \mu_1$.
From Little's theorem (Ch. 3.2 in~\cite{b:Bertsekas}), we have that
\begin{eqnarray} \label{eq:Prnonempty_fix}
\mathrm{Pr}\left(Q_1 > 0 \right) = \frac{\lambda_1}{\mathrm{Pr}\left(\mathcal{D}_{1/1} \right)}.
\end{eqnarray}

The service rate for the second queue is given by (\ref{eq:mu2}). After substituting (\ref{eq:Prnonempty_fix}) into (\ref{eq:mu2}) we obtain
\begin{eqnarray}
\mu_2 = \mathrm{Pr}\left(\mathcal{D}_{2/2} \right) + \frac{\mathrm{Pr}\left(\mathcal{D}_{2/1,2} \right) - \mathrm{Pr}\left(\mathcal{D}_{2/2} \right)}{\mathrm{Pr}\left(\mathcal{D}_{1/1} \right)}\lambda_1.
\end{eqnarray}

From Loynes' criterion we have that the second queue is stable if and only if $\lambda_2 < \mu_2$. The stability region for the channel is given then by (\ref{eq:Rfix}) and is depicted in Fig.~\ref{fig:region}.
\begin{align} \label{eq:Rfix}
\mathcal{R} = \left\lbrace (\lambda_{1},\lambda_{2}): \frac{\lambda_2}{\mathrm{Pr}\left(\mathcal{D}_{2/2} \right)} + \frac{\mathrm{Pr}\left(\mathcal{D}_{2/2} \right) - \mathrm{Pr}\left(\mathcal{D}_{2/1,2} \right)}{\mathrm{Pr}\left(\mathcal{D}_{1/1} \right)\mathrm{Pr}\left(\mathcal{D}_{2/2} \right)}\lambda_1 < 1  , \lambda_1 < \mathrm{Pr} \left(\mathcal{D}_{1/1} \right)  \right\rbrace.
\end{align}

Recall that the success probability $\mathrm{Pr}\left(\mathcal{D}_{2/1,2} \right)$ is given by (\ref{eq:Pr212}).
Note that in this case we do not face the problem of coupled queues as mentioned in the general case described in Section \ref{sec:stability_general}.

\subsubsection{The case where $\gamma_2 P_1 < P_2 \leq  P_1 \frac{\gamma_2 (1+\gamma_1)}{\gamma_1}$}
In this case clearly $\mathrm{Pr}\left(\mathcal{D}_{1/1,2} \right) \neq \mathrm{Pr} \left(\mathcal{D}_{1/1} \right)$. We obtain that
\begin{equation} \label{eq:P112_case1}
\mathrm{Pr}\left(\mathcal{D}_{1/1,2} \right) = \exp \left(- \frac{\gamma_2 d_1 ^{\alpha}}{ P_2 -\gamma_2 P_1 } \right).
\end{equation}

In this case the queues are coupled so we have to use the results obtained in Section \ref{sec:stability_general} derived using the stochastic dominance technique. The stability region is obtained by using (\ref{eq:P112_case1}) and (\ref{eq:Pr212}) into (\ref{eq:R_1}) and (\ref{eq:R_2}).

\begin{figure}[]
\centering
\includegraphics[scale=0.5]{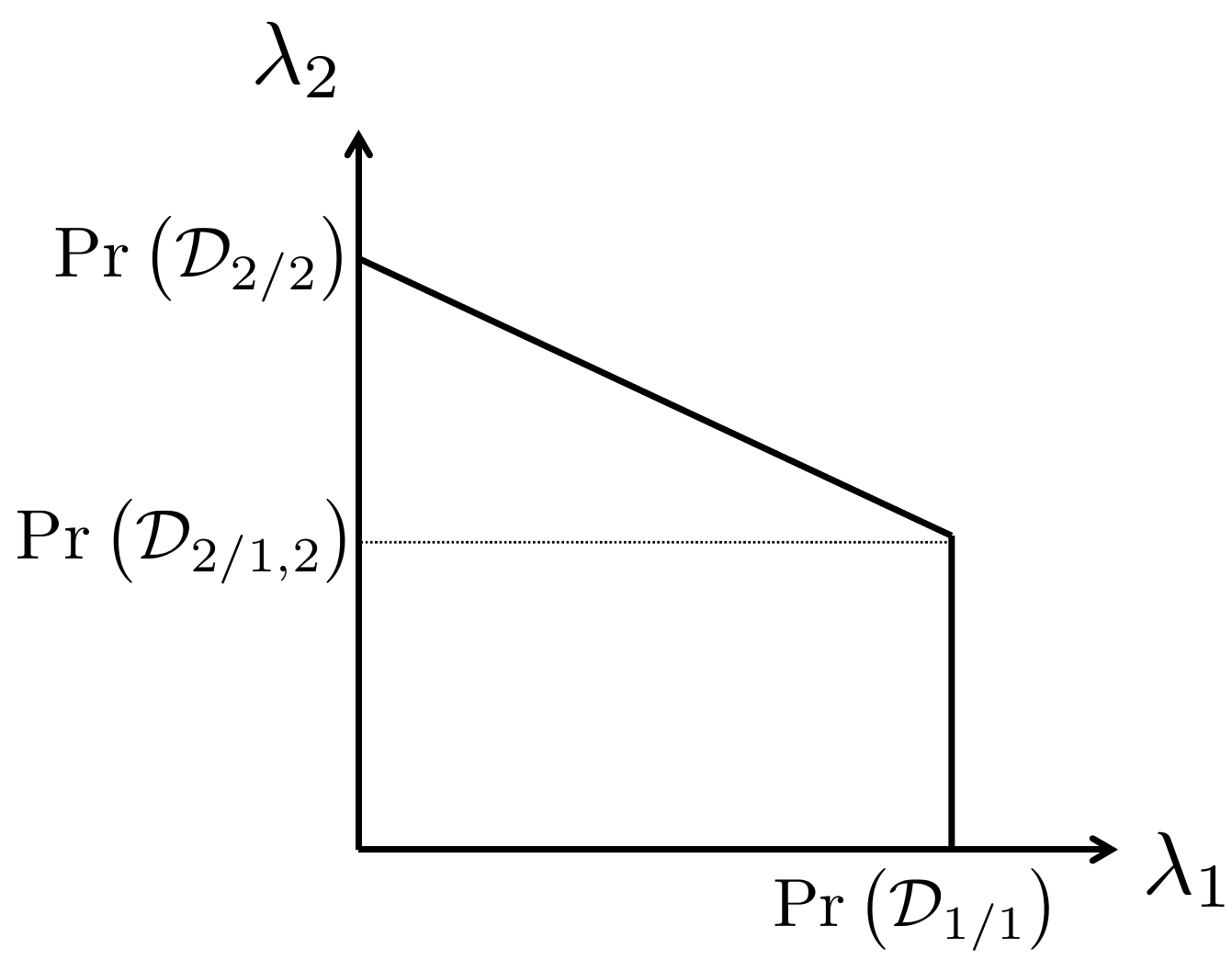}
\caption{The stability region for fixed transmit powers when $P_2 > P_1 \frac{\gamma_2 (1+\gamma_1)}{\gamma_1}$ and $D_1$ applies successive decoding and $D_2$ treats interference as noise. In this case the region is always convex.}
\centering
\label{fig:region}
\end{figure}

\subsection{Variable Power Scheme based on Queue State}
In this part, we consider a simple adaptive scheme regarding the power allocated to the packets of each receiver. 
The power allocation is performed as described earlier, namely: when both queues are not empty, the transmit power for the first and second queue is $P_1$ and $P_2$, respectively, satisfying $P_1+P_2=P$. However, when the queue of $i$-th receiver is empty, the total transmit power $P$ is used for transmitting the packets for the $j$-th (where $j \neq i$) receiver.

The average service rates of the first and the second queue, $\mu_1$ and $\mu_2$ are given by (\ref{eq:mu1}) and (\ref{eq:mu2}) respectively.
The success probabilities $\mathrm{Pr}\left(\mathcal{D}_{i/i} \right)$ for $i=1,2$ are given by
\begin{eqnarray} \label{eq:Prii}
\mathrm{Pr}\left(\mathcal{D}_{i/i} \right) = \exp \left( - \frac{\gamma_i d_i^{\alpha}}{P} \right),
\end{eqnarray}
since when a queue is empty, the transmitter assigns all power to the other queue, and can be obtained from (\ref{eq:SNR}). The success probability $\mathrm{Pr}\left(\mathcal{D}_{1/1,2} \right)$ is given by (\ref{eq:Pr112_SC}).

In the above scheme, it is evident that $\mathrm{Pr}\left(\mathcal{D}_{1/1} \right) \neq \mathrm{Pr}\left(\mathcal{D}_{1/1,2} \right)$, and as a result, there is coupling between the queues. Thus, we can use directly the stability region obtained in Section \ref{sec:stability_general} by substituting the values of the success probabilities.

The stability region $\mathcal{R}$ has two parts, $\mathcal{R}_1$ and $\mathcal{R}_2$, where $\mathcal{R} = \mathcal{R}_1 \bigcup \mathcal{R}_2$.
If power $\gamma_2 P_1 < P_2 \leq  P_1 \frac{\gamma_2 (1+\gamma_1)}{\gamma_1}$, the region is given by
$\mathcal{R} = \mathcal{R}_1 \bigcup \mathcal{R}_2$, where $\mathcal{R}_1$ and $\mathcal{R}_2$ are given by (\ref{eq:R_1_changing2}) and (\ref{eq:R_2_changing2}) respectively.
\begin{align}
\label{eq:R_1_changing2}
\mathcal{R}_1= \left\lbrace  (\lambda_{1},\lambda_{2}): \frac{\lambda_1}{\exp \left( - \frac{\gamma_1 d_1^{\alpha}}{P} \right)} + \frac{\exp \left( - \frac{\gamma_1 d_1^{\alpha}}{P} \right) - \exp \left(- \frac{\gamma_2 d_1 ^{\alpha}}{ P_2 -\gamma_2 P_1 } \right)}{\exp \left( - \frac{\gamma_1 d_1^{\alpha}}{P} \right) \exp \left(- \frac{\gamma_2 d_2 ^{\alpha}}{ (1+\gamma_2) P_2 - \gamma_2 } \right)}\lambda_2 < 1, 
 \lambda_2 <  \exp \left(- \frac{\gamma_2 d_2 ^{\alpha}}{ (1+\gamma_2) P_2 - \gamma_2 } \right) \right\rbrace.
\end{align}
\begin{align}
\label{eq:R_2_changing2}
\mathcal{R}_2= \left\lbrace  (\lambda_{1},\lambda_{2}): \frac{\lambda_2}{\exp \left( - \frac{\gamma_2 d_2^{\alpha}}{P} \right)} + \frac{\exp \left( - \frac{\gamma_2 d_2^{\alpha}}{P} \right) - \exp \left( - \frac{\gamma_2 d_2^{\alpha}}{(1+\gamma_2) P_2 - \gamma_2} \right)}{\exp \left( - \frac{\gamma_2 d_2^{\alpha}}{P} \right) \exp \left(- \frac{\gamma_2 d_1 ^{\alpha}}{ P_2 -\gamma_2 P_1 } \right)}\lambda_2 < 1, \lambda_1 < \exp \left(- \frac{\gamma_2 d_1 ^{\alpha}}{ P_2 -\gamma_2 P_1 } \right) \right\rbrace.
\end{align}
If $P_2 > P_1 \frac{\gamma_2 (1+\gamma_1)}{\gamma_1}$, then $\mathcal{R}_1$ is given in (\ref{eq:R_1_changing}) after using (\ref{eq:Pr212}), (\ref{eq:Pr112_SC}) and (\ref{eq:Prii}) along with (\ref{eq:R_1}).
\begin{align}
\label{eq:R_1_changing}
\mathcal{R}_1= \left\lbrace  (\lambda_{1},\lambda_{2}): \frac{\lambda_1}{\exp \left( - \frac{\gamma_1 d_1^{\alpha}}{P} \right)} + \frac{\exp \left( - \frac{\gamma_1 d_1^{\alpha}}{P} \right) - \exp \left( - \frac{\gamma_1 d_1^{\alpha}}{P_1} \right)}{\exp \left( - \frac{\gamma_1 d_1^{\alpha}}{P} \right) \exp \left(- \frac{\gamma_2 d_2 ^{\alpha}}{ (1+\gamma_2) P_2 - \gamma_2 } \right)}\lambda_2 < 1, 
 \lambda_2 <  \exp \left(- \frac{\gamma_2 d_2 ^{\alpha}}{ (1+\gamma_2) P_2 - \gamma_2 } \right) \right\rbrace.
\end{align}
Note that if $P_2 < \gamma_2 P_1$ the first dominant system is unstable and the region is empty, i.e. $\mathcal{R}_1 = \emptyset$.
Similarly $\mathcal{R}_2$ is given by (\ref{eq:R_2_changing}). 
\begin{align}
\label{eq:R_2_changing}
\mathcal{R}_2= \left\lbrace  (\lambda_{1},\lambda_{2}): \frac{\lambda_2}{\exp \left( - \frac{\gamma_2 d_2^{\alpha}}{P} \right)} + \frac{\exp \left( - \frac{\gamma_2 d_2^{\alpha}}{P} \right) - \exp \left( - \frac{\gamma_2 d_2^{\alpha}}{(1+\gamma_2) P_2 - \gamma_2} \right)}{\exp \left( - \frac{\gamma_2 d_2^{\alpha}}{P} \right) \exp \left(- \frac{\gamma_1 d_1 ^{\alpha}}{ P_1 } \right)}\lambda_2 < 1, \lambda_1 <  \exp \left(- \frac{\gamma_1 d_1 ^{\alpha}}{ P_1 } \right) \right\rbrace.
\end{align}
The indistinguishability argument mentioned in Section \ref{sec:stability_general} applies to this case as well.

\begin{remark}
The case where both receivers apply successive decoding can also be derived from the results above; however, it is omitted since it does not provide better results than the case we consider. Additionally, it can even yield poorer results than in the case where both receivers treat interference as noise.
\end{remark}

\section{Numerical Results} \label{sec:NumResults}
In this section, we evaluate numerically the analytical results obtained in the previous sections.
We let $P_1=80$, $P_2=120$, $P =200$, $\alpha=2$, $d_1=10$ and $d_2=13$. Tables \ref{tab:succ_probs} and \ref{tab:succ_probs_interference} summarize the values of success probabilities under different schemes and various values for $\gamma_1$ and $\gamma_2$. TIN denotes the case that both users treat interference as noise and SD denotes the case where the first user applies successive decoding. Note that for SD we have that $\mathrm{Pr}\left(\mathcal{D}_{1/1,2} \right) = \mathrm{Pr}\left(\mathcal{D}_{1/1} \right)$, because condition (\ref{eq:condition_1stRx}) holds. PC denotes the case where the variable power scheme based on queue state is used.

\begin{table}[ht]
\begin{center}
\begin{tabular}{|c|c|c||c|c|}
\hline
$\gamma_1$ & $\gamma_2$ & Power & $\mathrm{Pr}\left(\mathcal{D}_{1/1} \right)$ & $\mathrm{Pr}\left(\mathcal{D}_{2/2} \right)$ \\\hline\hline
0.5 & 0.4 & $P_1=80$, $P_2=120$ &  0.5353 &  0.5203  \\ \hline
0.5 & 0.4 & PC ($P_1=P_2=200$) &  0.7788 & 0.6757 \\ \hline
1.2 & 0.7 & $P_1=80$, $P_2=120$ &  0.2231 & 0.3188   \\ \hline
1.2 & 0.7 & PC ($P_1=P_2=200$) & 0.5488 & 0.5036  \\ \hline
\end{tabular}
\end{center}
\caption{$\mathrm{Pr}\left(\mathcal{D}_{1/1} \right)$ and $\mathrm{Pr}\left(\mathcal{D}_{2/2} \right)$ for various values of $\gamma_1$ and $\gamma_2$ for $d_1=10$, $d_2=14$ and $\alpha=2$.}
\label{tab:succ_probs}
\end{table}

\begin{table}[ht]
\begin{center}
\begin{tabular}{|c|c|c||c|c|}
\hline
$\gamma_1$ & $\gamma_2$ & Scheme & $\mathrm{Pr}\left(\mathcal{D}_{1/1,2} \right)$ & $\mathrm{Pr}\left(\mathcal{D}_{2/1,2} \right)$ \\\hline\hline
0.5 & 0.4 & TIN &  0.0821 &  0.4103  \\ \hline
0.5 & 0.4 & SD &  0.5353 & 0.4103  \\ \hline
1.2 & 0.7 & TIN &  0 & 0.1172   \\ \hline
1.2 & 0.7 & SD & 0.2231 & 0.1172  \\ \hline
\end{tabular}
\end{center}
\caption{$\mathrm{Pr}\left(\mathcal{D}_{1/1,2} \right)$ and $\mathrm{Pr}\left(\mathcal{D}_{2/1,2} \right)$ for various values of $\gamma_1$ and $\gamma_2$ for $P_1=80$, $P_2=120$, $d_1=10$, $d_2=14$ and $\alpha=2$.}
\label{tab:succ_probs_interference}
\end{table}

\subsection{Stability Region}
We provide here numerical results for the analysis in Sections \ref{sec:stability_IAN} and \ref{sec:stability_SC}. 
The stability region for $\gamma_1=0.5$, $\gamma_2=0.4$ is depicted in Fig. \ref{fig:STR_1}. SD-PC has the best performance, in the low $\lambda_1$ regime, while in the high $\lambda_2$ regime, TIN-PC has better performance than SD. This was expected because of the power control mechanism that allows to allocate the whole power to the non-empty queue when one queue is empty.

\begin{figure}[]
\centering
\includegraphics[scale=0.75]{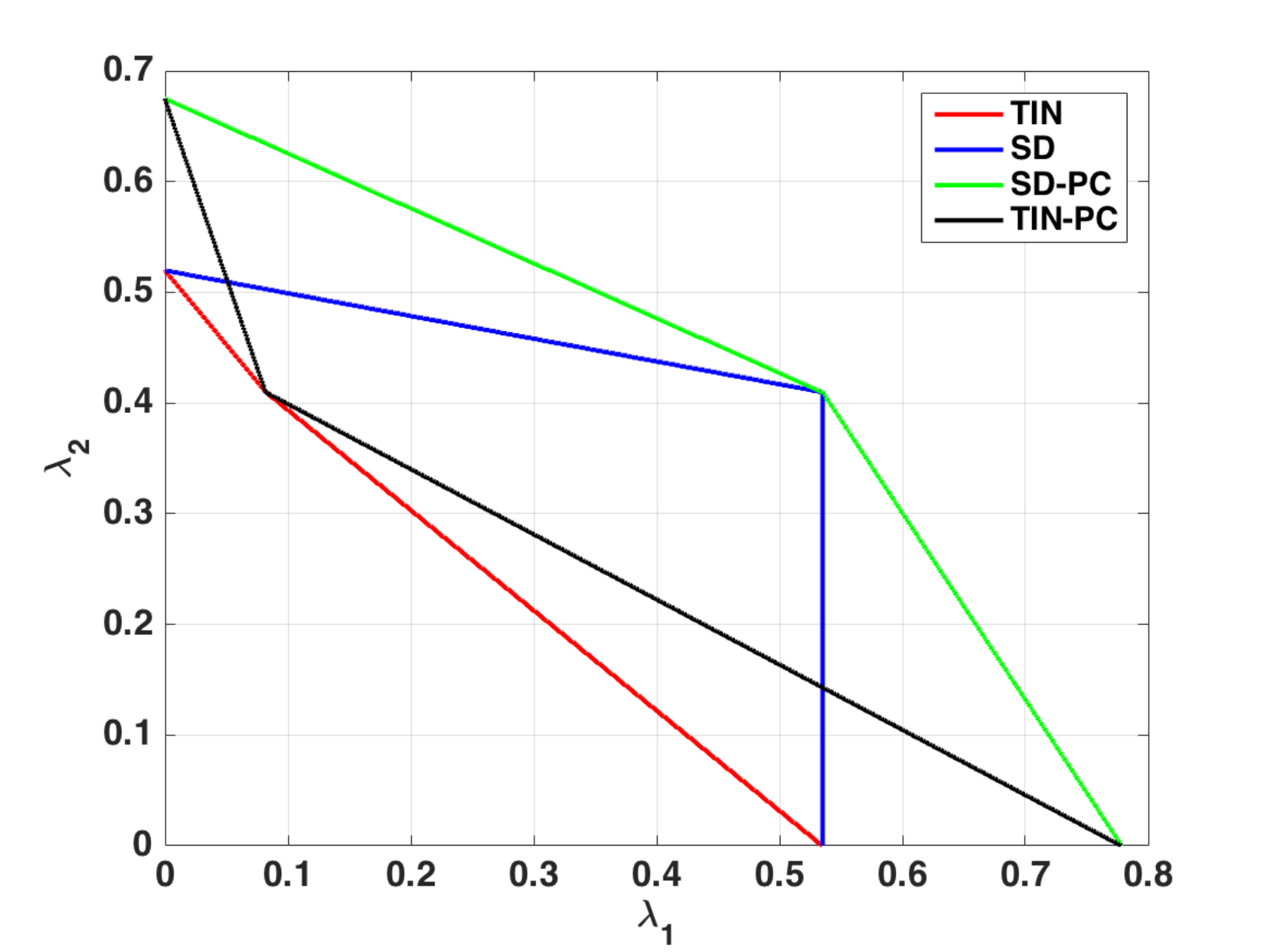}
\caption{The stability region for $\gamma_1=0.5$ and $\gamma_2=0.4$ when $P_1=80$ and $P_2=120$.}
\centering
\label{fig:STR_1}
\end{figure}

The stability region for $\gamma_1=1.2$, $\gamma_2=0.7$ is depicted in Fig. \ref{fig:STR_2}. Note that SD-PC has the best performance; TIN-PC has better performance for the low-$\lambda_2$ and high-$\lambda_1$ regimes while SD is better for low-$\lambda_1$ and high-$\lambda_2$ regimes. 

\begin{figure}[]
\centering
\includegraphics[scale=0.75]{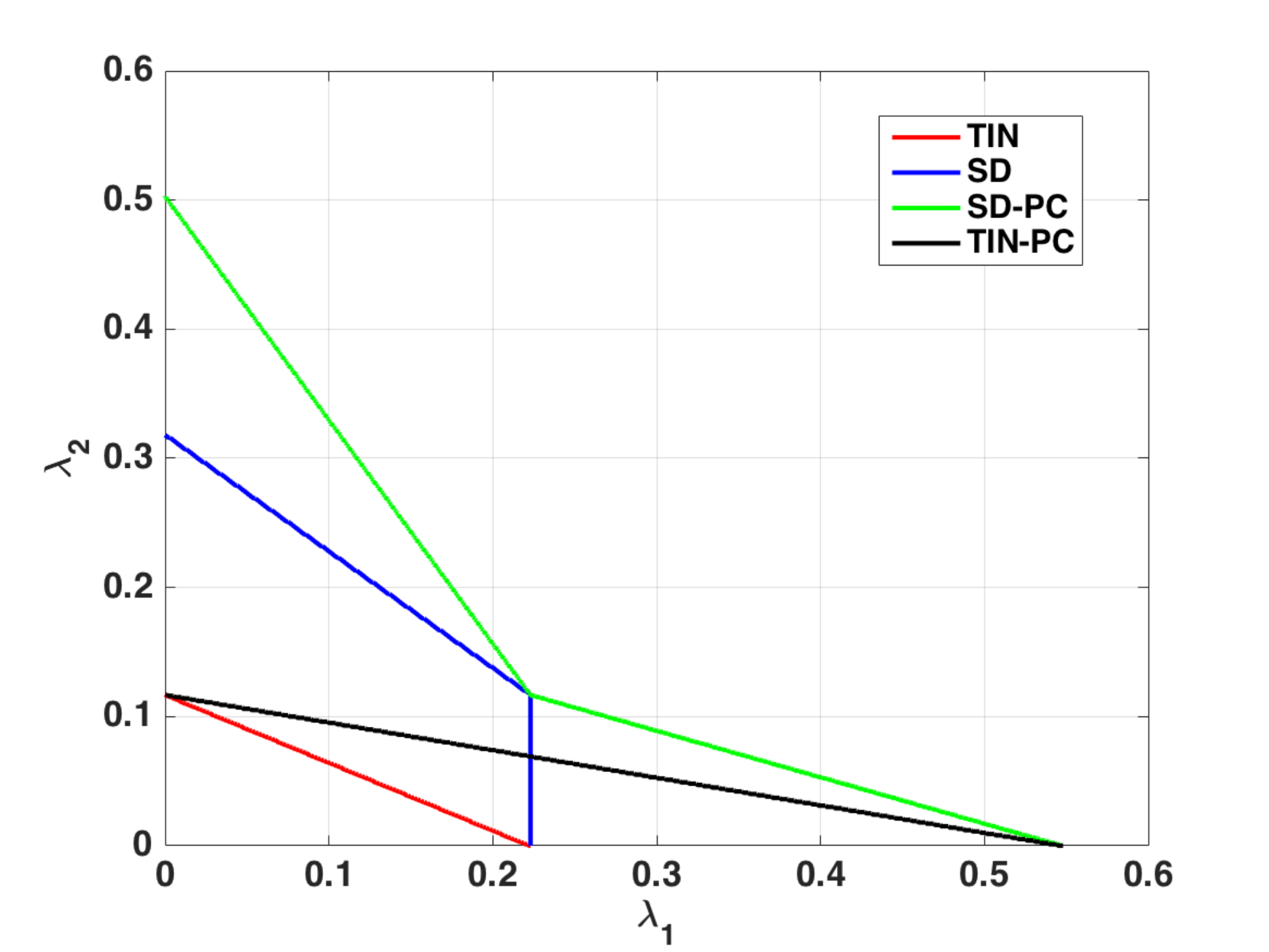}
\caption{The stability region for $\gamma_1=1.2$ and $\gamma_2=0.7$ when $P_1=80$ and $P_2=120$.}
\centering
\label{fig:STR_2}
\end{figure}
In both cases, SD-PC can support higher arrival rates than TIN.

\subsection{Closure of the Stability Region}
In this subsection, we numerically evaluate the closure of the stability
region as discussed in Section \ref{sec:closure}. The figures in this subsection are approximations of the closure because they obtained by plotting the stability region for different power values. We consider two cases for the SINR thresholds $\gamma_1=0.5$, $\gamma_2=0.4$ and $\gamma_1=1.2$, $\gamma_2=0.7$.

The closure of the stability region for $\gamma_1=0.5$, $\gamma_2=0.4$ is depicted in Fig. \ref{fig:Closure_Case1}. The impact of the power control policy is noticeable in both TIN and SD cases. The best performance (larger region) is achieved by using SD-PC. Both SD-PC and TIN-PC have convex regions, which means that the performance is superior to time-sharing schemes. However, the region for TIN-PC has a triangular shape, which means that its performance is very close to that of time-sharing. In contrast, the region in both TIN and SD cases is non-convex. In the SD case, for large values of $\lambda_1$ and small values of $\lambda_2$ there is a part missing when we compare to SD-PC, which destroys the convexity. The TIN case has worse performance than time-sharing.

\begin{figure}[htb]
\centering
  \subfloat[TIN]{%
    \includegraphics[scale=0.45]{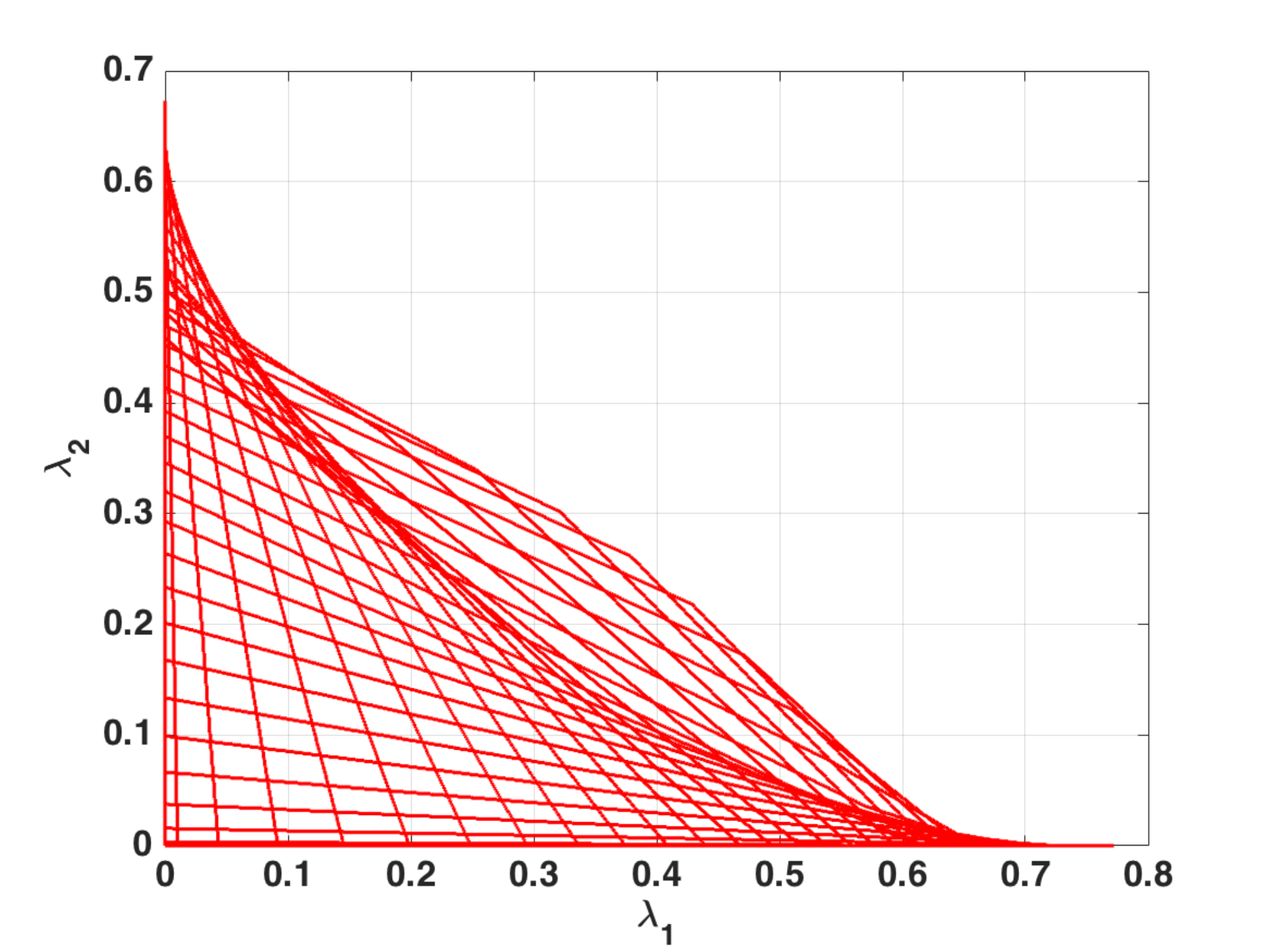}}
  \subfloat[TIN-PC]{%
    \includegraphics[scale=0.45]{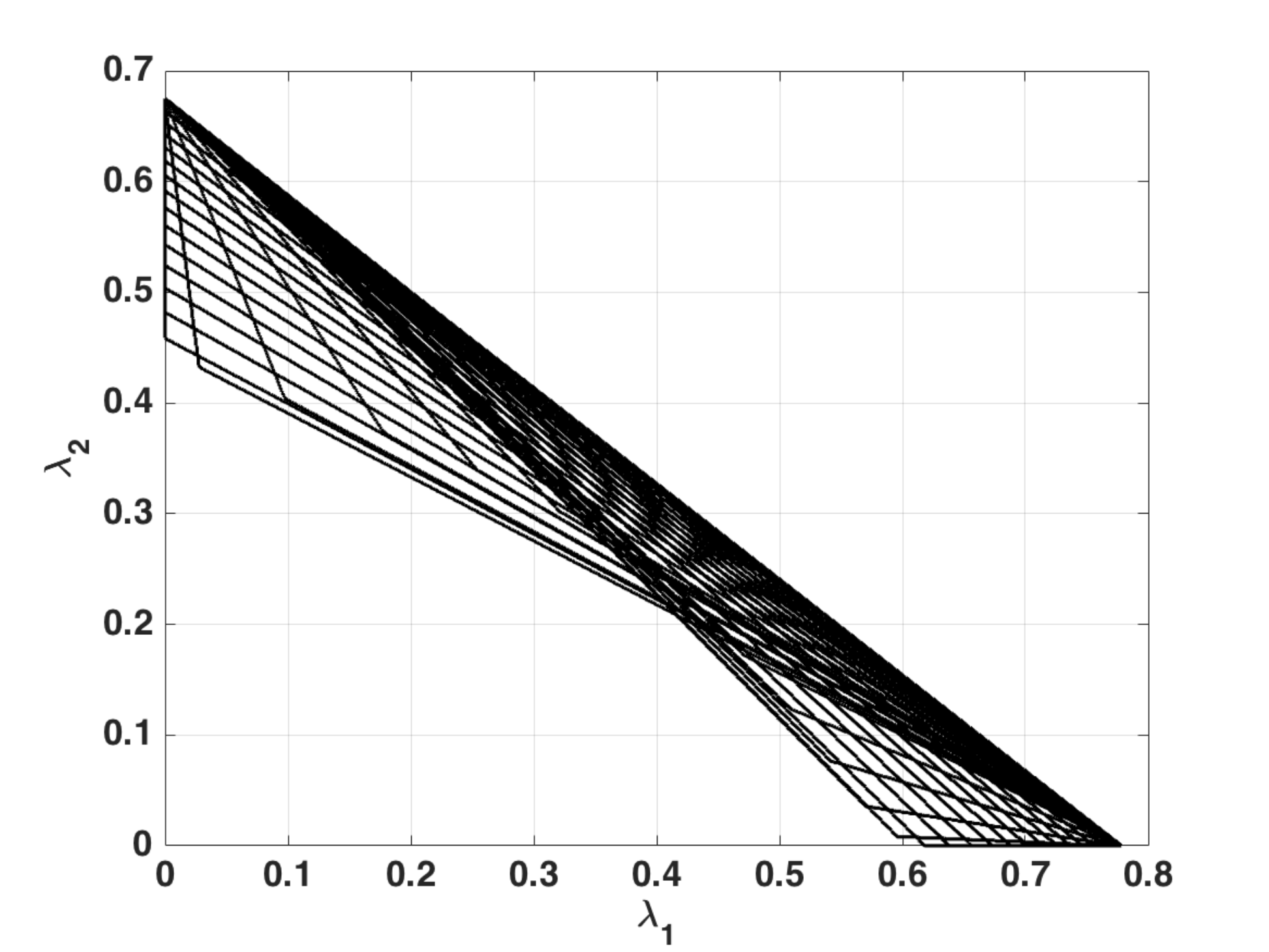}} \\
  \subfloat[SD]{%
    \includegraphics[scale=0.45]{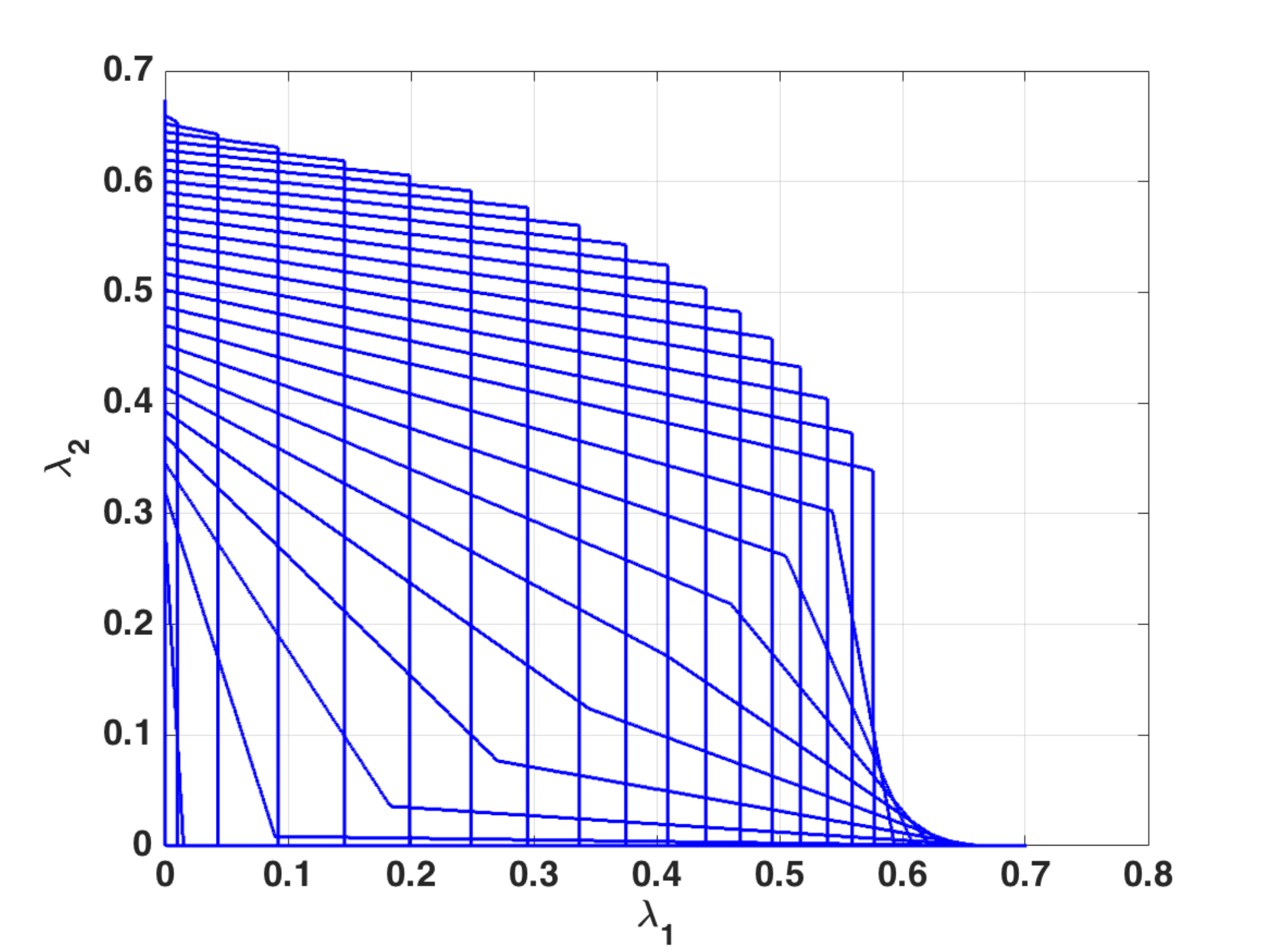}}
  \subfloat[SD-PC]{%
    \includegraphics[scale=0.45]{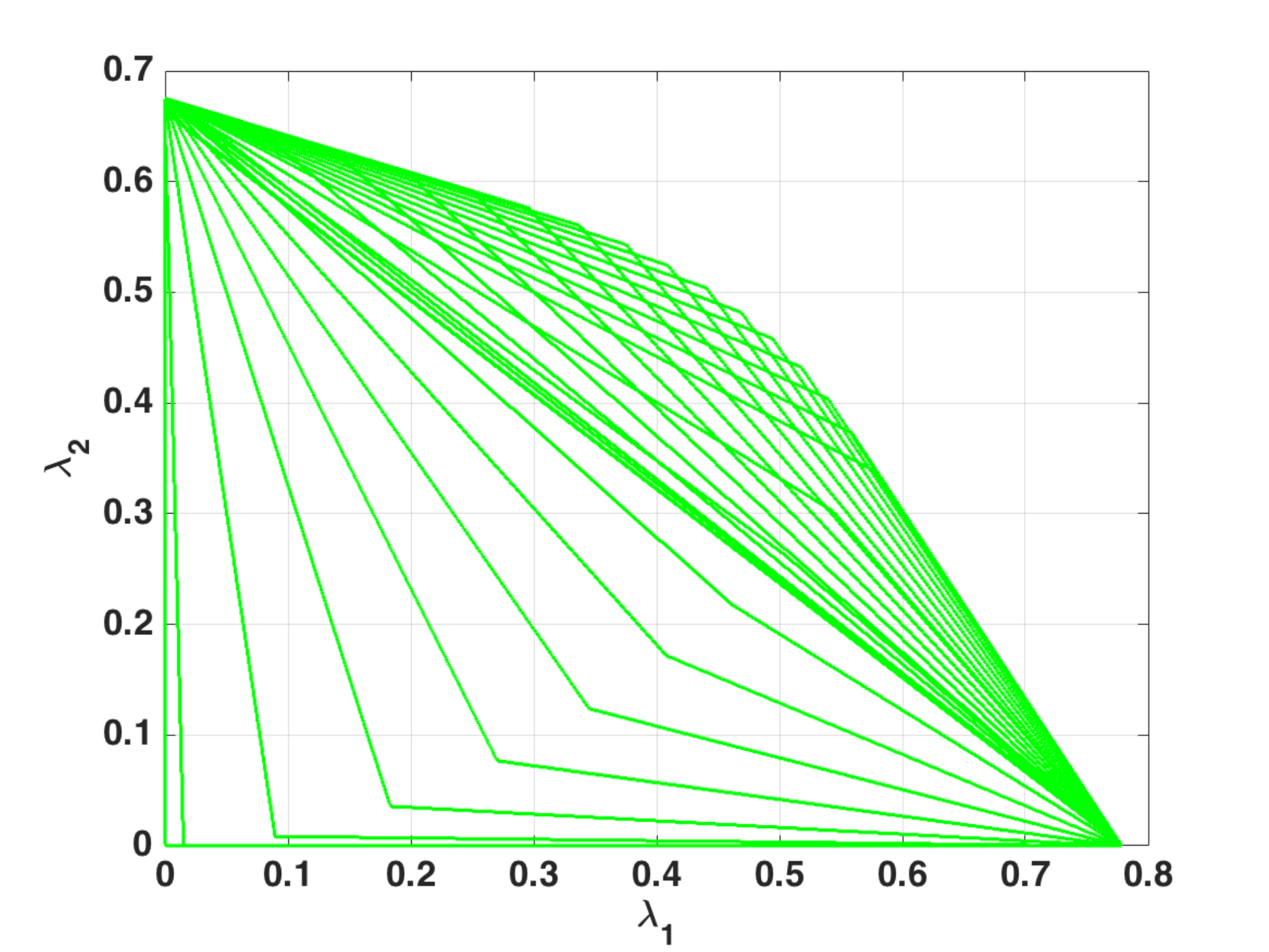}}\\
  \caption{The closure of the stability region for the case where $\gamma_1=0.5$ and $\gamma_2=0.4$.}\label{fig:Closure_Case1}
\end{figure}

The closure of the stability region for $\gamma_1=1.2$, $\gamma_2=0.7$ is depicted in Fig. \ref{fig:Closure_Case2}. Since $\gamma_1=1.2$, interference is more destructive for the first user. This results in TIN-PC and SD-PC having similar performance, and close to that of time-sharing. The performance of TIN and SD is worse than the respective schemes with power control; and the closures of the stability regions are non-convex.

\begin{figure}[htb]
\centering
  \subfloat[TIN]{%
    \includegraphics[scale=0.45]{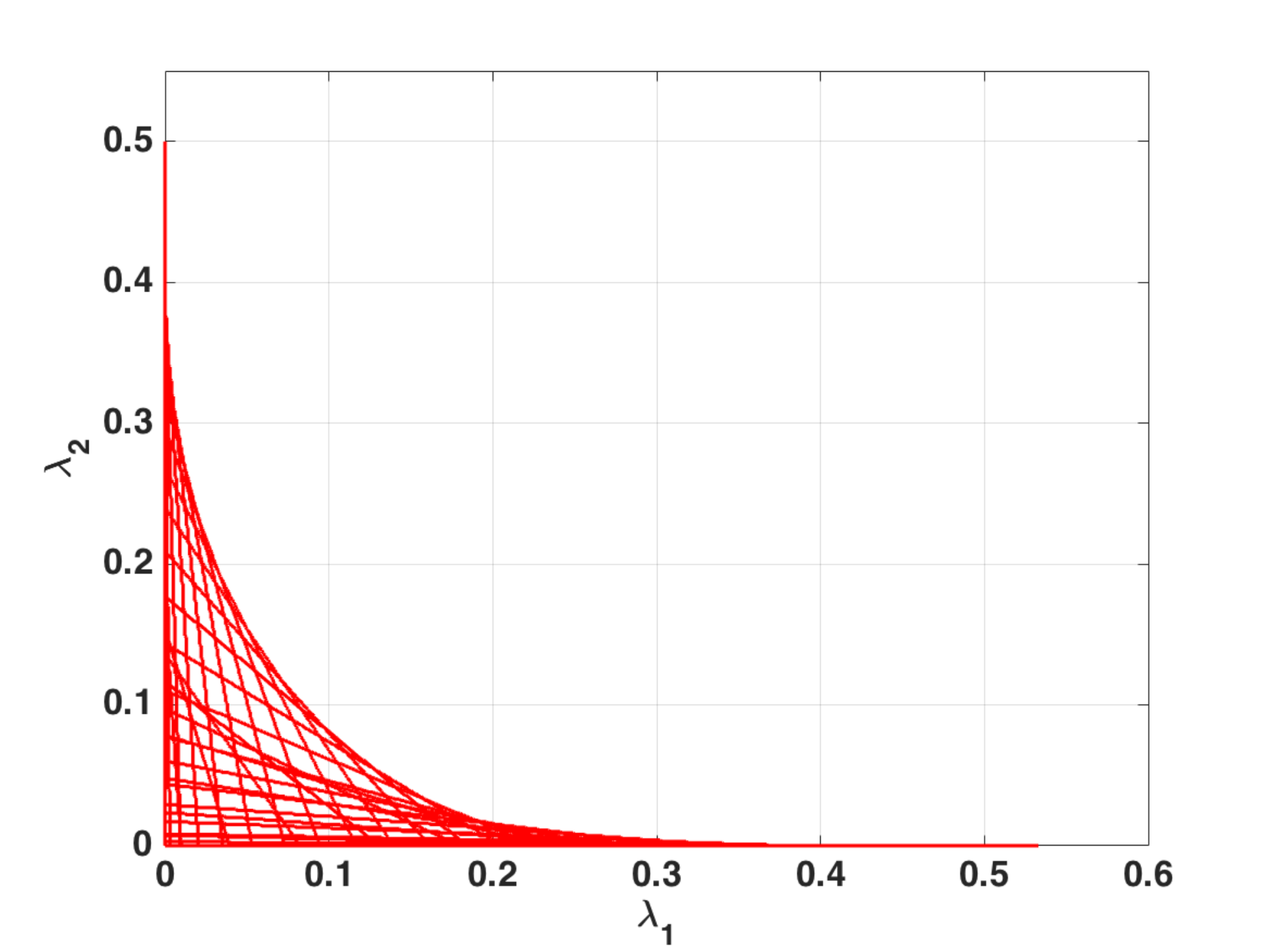}}
  \subfloat[TIN-PC]{%
    \includegraphics[scale=0.45]{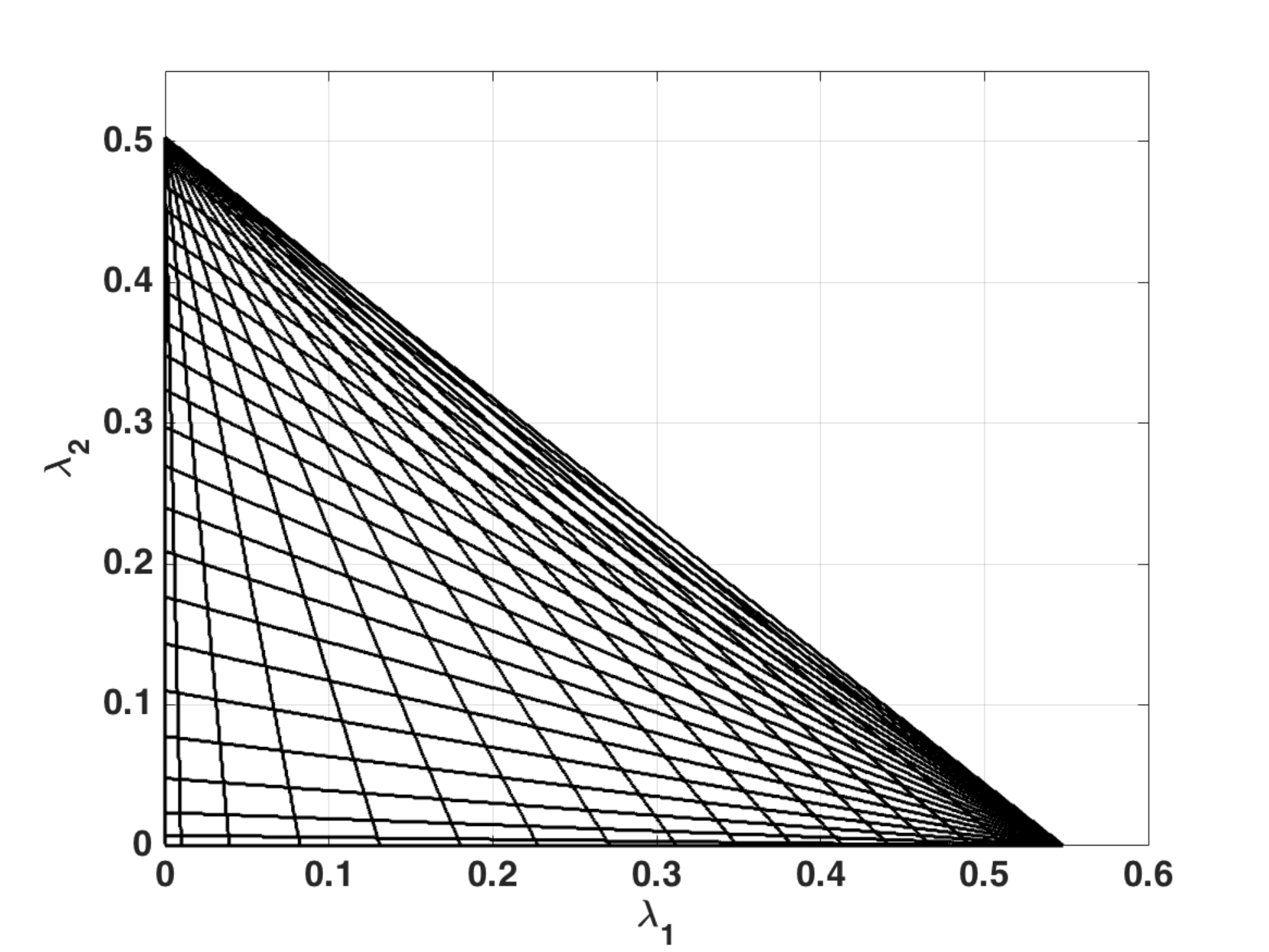}} \\
  \subfloat[SD]{%
    \includegraphics[scale=0.45]{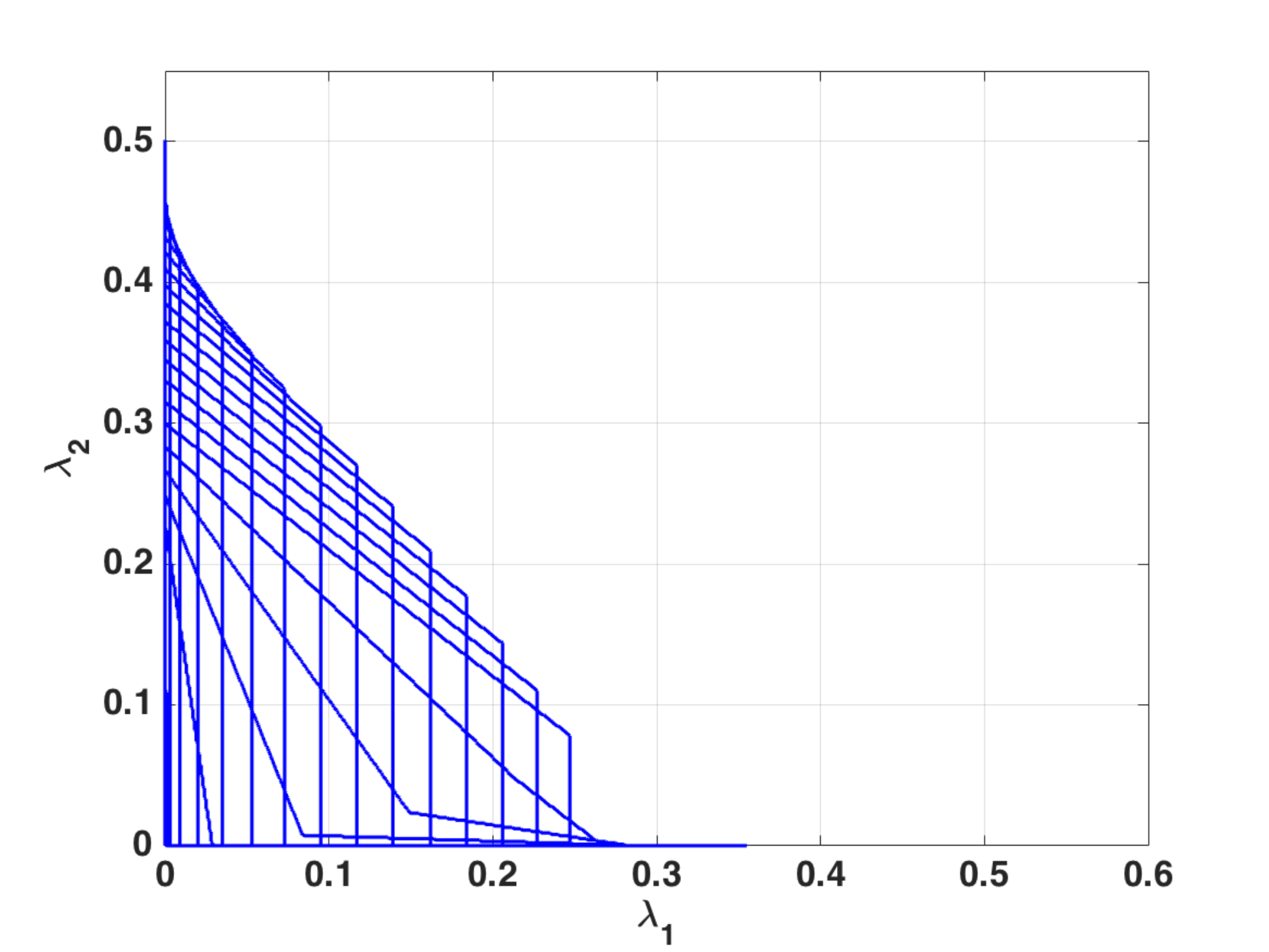}}
  \subfloat[SD-PC]{%
    \includegraphics[scale=0.45]{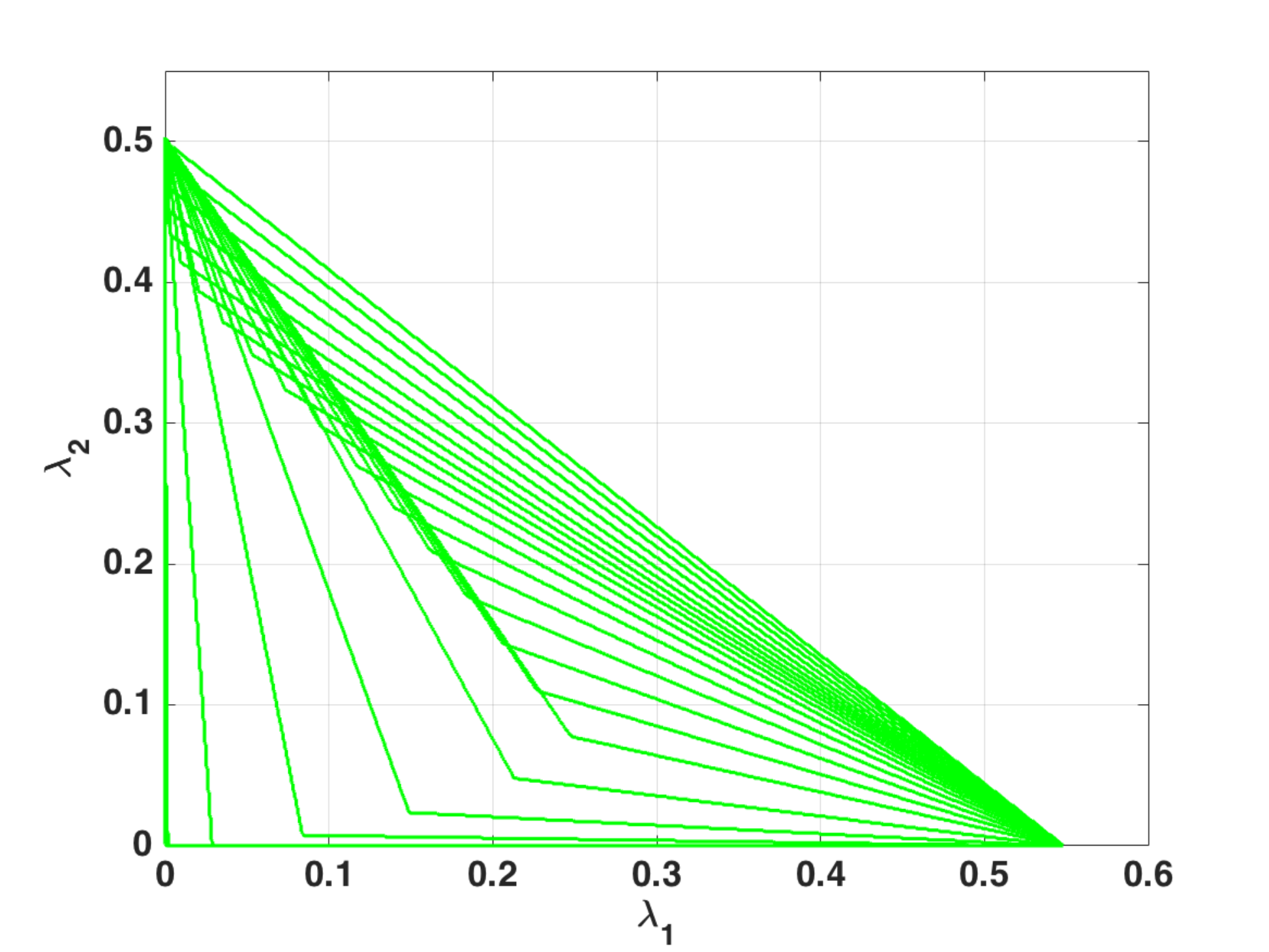}}\\
  \caption{The closure of the stability region for the case where $\gamma_1=1.2$ and $\gamma_2=0.7$.}\label{fig:Closure_Case2}
\end{figure}

The main message from this set of numerical evaluations is that power control has a significant impact on the closure of the stable throughput region, even in the case where both receivers treat interference as noise. 

\subsection{Maximum Aggregate Stable Throughput}
We numerically evaluate here the maximum aggregate stable throughput as discussed in Section \ref{sec:Aggregate}. We also compare the maximum aggregate stable throughput with the maximum aggregate saturated throughput where the queues of the source are always backlogged. In the latter, the aggregate saturated throughput is $T_\mathrm{aggr,SAT}= \mathrm{Pr}\left(\mathcal{D}_{1/1,2} \right)+\mathrm{Pr}\left(\mathcal{D}_{2/1,2} \right)$. Regarding the maximum aggregate stable throughput, we only consider the cases with power control as they are more interesting given their superior performance.

The case for $\gamma_1=0.5$, $\gamma_2=0.4$ is depicted in Fig. \ref{fig:Aggr_Case1}.
We observe that in the low or in the high power regime, the achieved maximum aggregate stable throughput is $\mathrm{Pr}\left(\mathcal{D}_{1/1} \right)$ for both SD-PC and TIN-PC, as it can also be observed in Tables \ref{tab:succ_probs} and \ref{tab:succ_probs_interference}. However, in the intermediate power regime, the maximum is $\mathrm{Pr}\left(\mathcal{D}_{1/1,2} \right) + \mathrm{Pr}\left(\mathcal{D}_{2/1,2} \right)$ and it is the same with SD-SAT. This maximum is achieved when both the queues are active meaning that the system can sustain parallel transmissions without performance degradation. 

\begin{figure}[]
\centering
\includegraphics[scale=0.75]{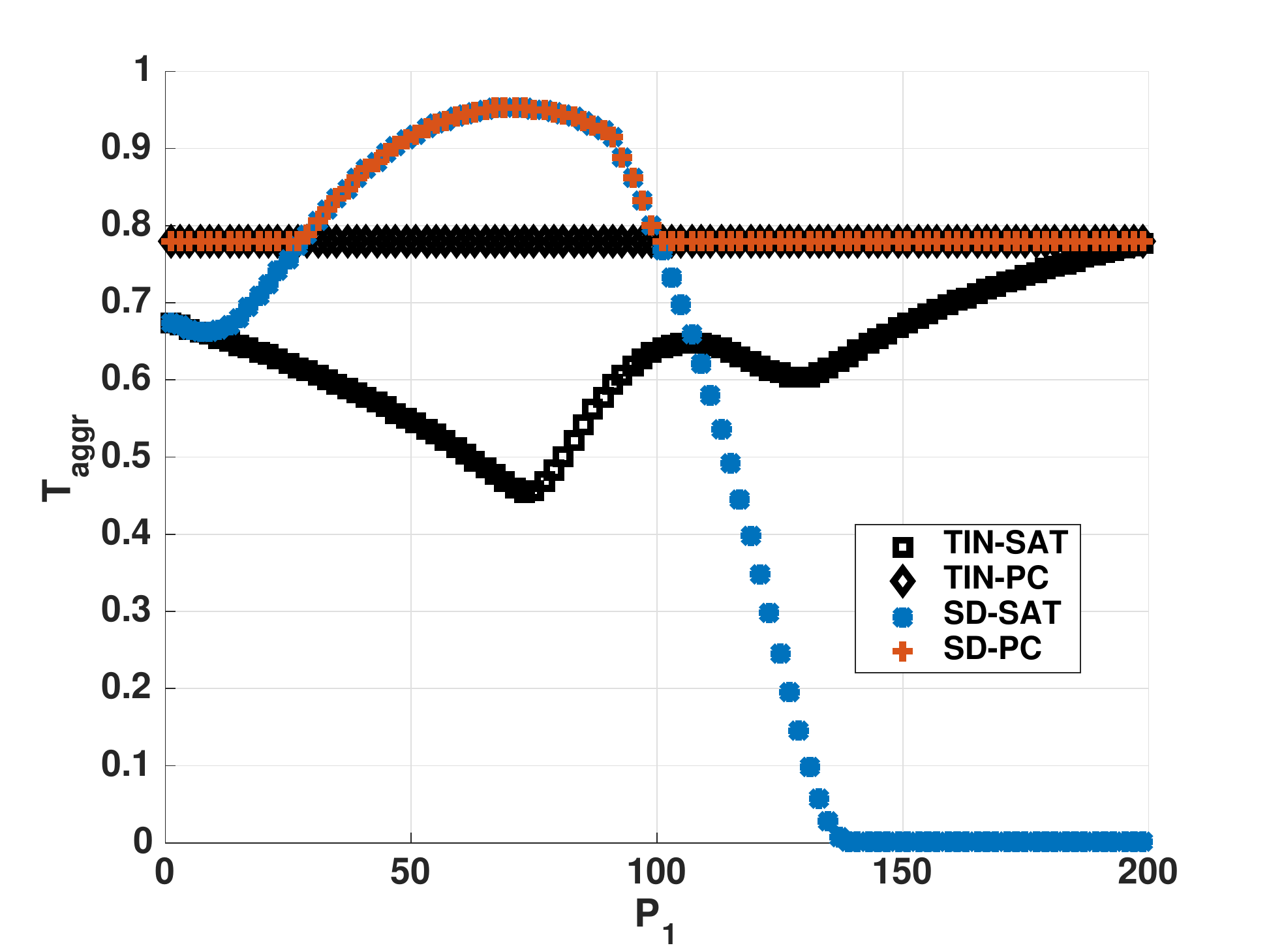}
\caption{The maximum aggregate stable throughput vs $P_1$ for $\gamma_1=0.5$ and $\gamma_2=0.4$. The maximum aggregate saturated throughput is depicted and denoted by SAT for each scheme.}
\centering
\label{fig:Aggr_Case1}
\end{figure}

The case for $\gamma_1=1.2$, $\gamma_2=0.7$ is depicted in Fig. \ref{fig:Aggr_Case2}. Since $\gamma_1 > 1$, the achieved maximum aggregate stable throughput is $\mathrm{Pr}\left(\mathcal{D}_{1/1} \right)$ for both SD-PC and TIN-PC (see also Tables \ref{tab:succ_probs} and \ref{tab:succ_probs_interference}). The maximum aggregate saturated throughput is always less than the stable one.
This case implies that if we want to maximize the aggregate stable throughput, it is better to transmit on orthogonal channels for each user, rather than transmitting concurrently.

\begin{figure}[]
\centering
\includegraphics[scale=0.75]{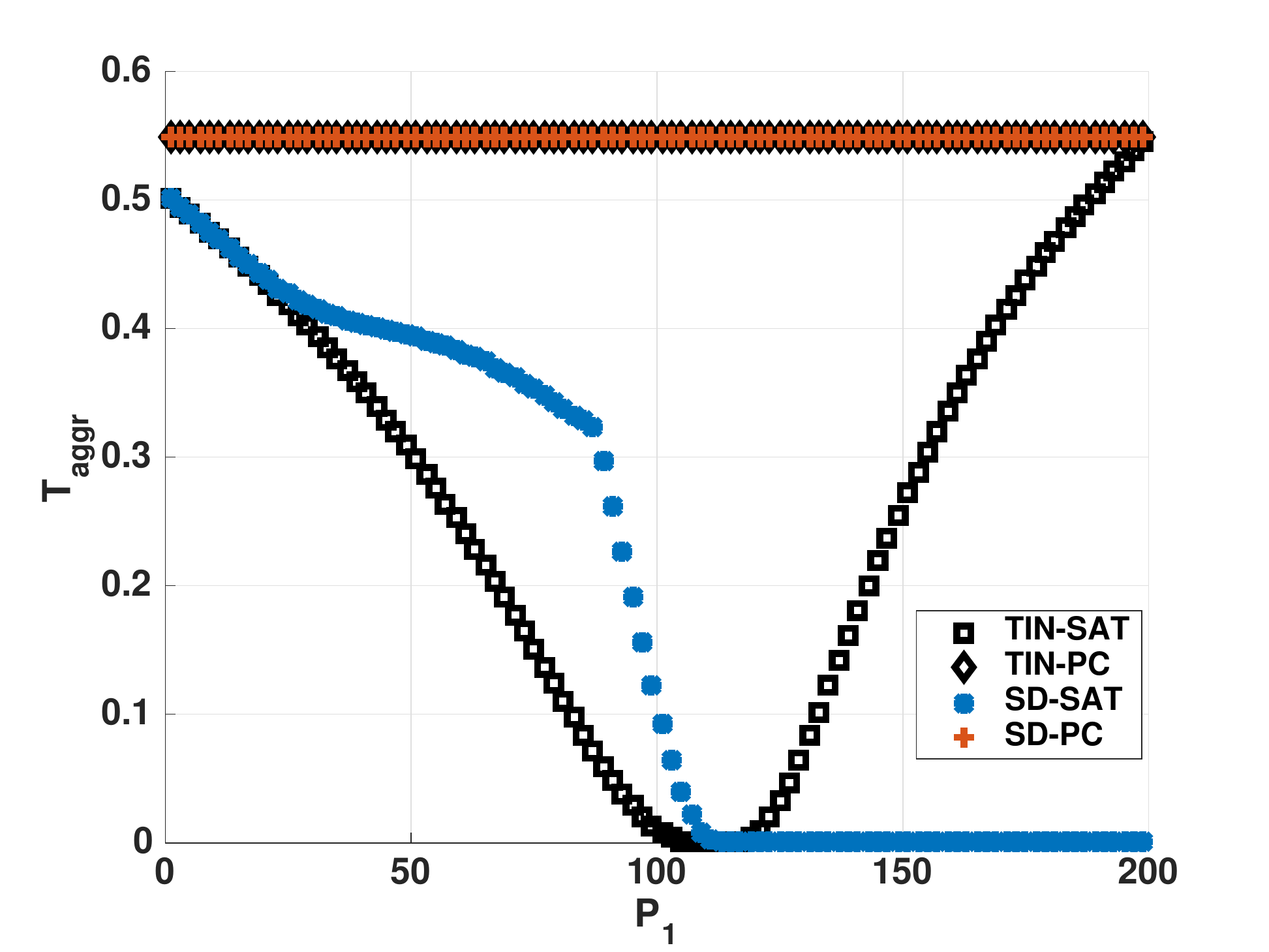}
\caption{The maximum aggregate stable throughput vs $P_1$ for $\gamma_1=1.2$ and $\gamma_2=0.7$. The maximum aggregate saturated throughput is depicted and denoted by SAT for each scheme.}
\centering
\label{fig:Aggr_Case2}
\end{figure}

\section{Interpretation of the Results}

In the previous section we presented numerical results of the stability region for various values of the parameters.

It is of interest to consider the relation between the different relevant rate measures such as the information-theoretic capacity region, the stable throughout region and the saturated throughput. These rates measures are derived under different assumptions and for subtly different models. More specifically, the capacity region is derived under the assumption that the users are backlogged and their queues are assumed saturated, which means that they are never empty and furthermore under elaborate coding procedures. The saturated throughput is expressed in packets per slot that are achievable in a network. The users are considered saturated here as well. The relationship among these regions is quite complex and not fully understood.

\begin{figure}[]
\centering
\includegraphics[scale=0.6]{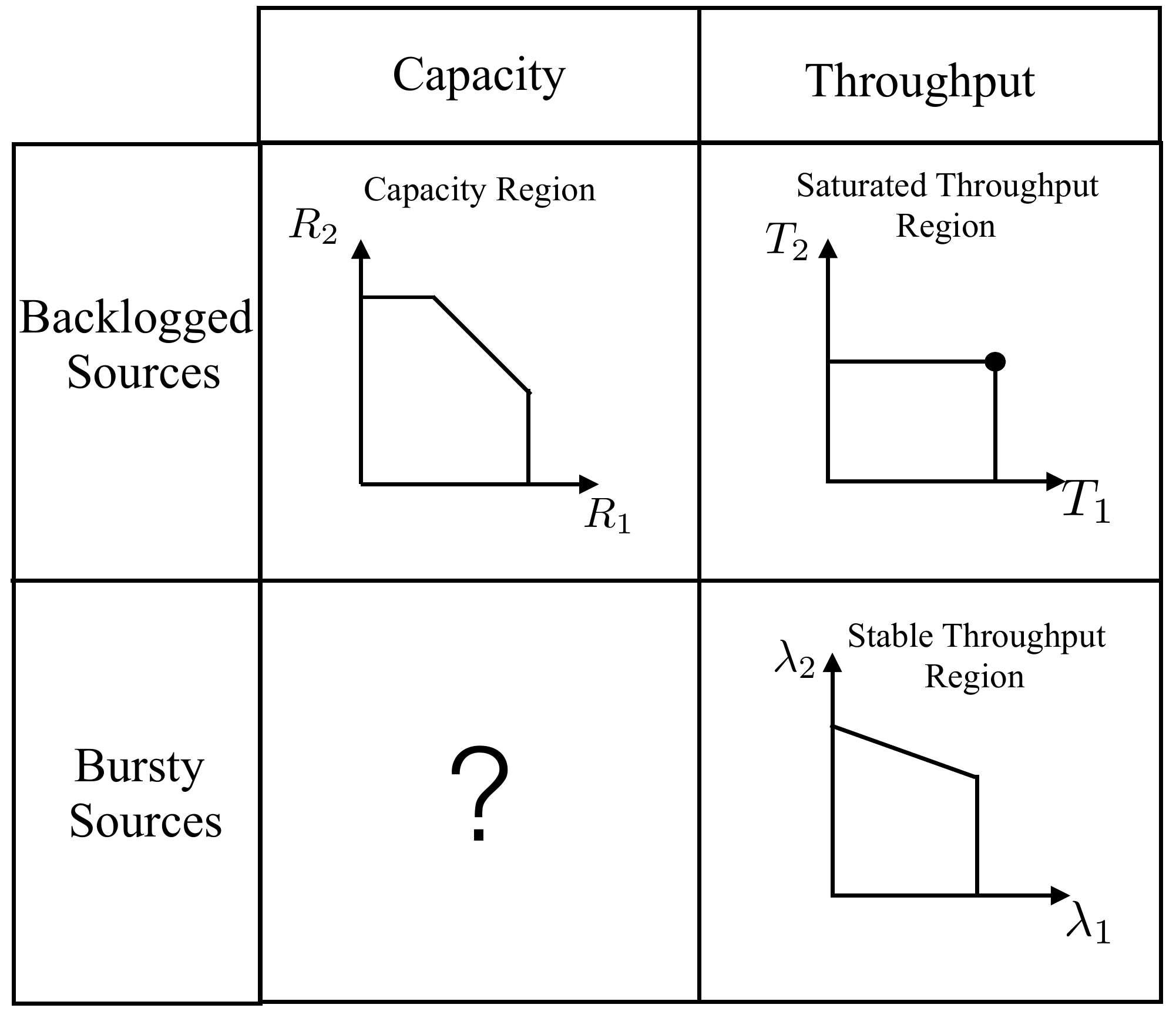}
\caption{The different regions for the two-user broadcast channel.}
\centering
\label{fig:regions}
\end{figure}

In Fig. \ref{fig:regions} we illustrate with a schematic the four possible rate regions that can be considered for every multi-user system.

There is a fourth metric which is the equivalent to the capacity region but for bursty sources. In general, we know very little about this region and it is of interest to consider the case of capacity under random arrivals at the sources. As seen in \cite{b:VerduTIT96}, a source with random arrivals can achieve higher rates than the saturated case. 

As stable throughput region can be an outer bound of the saturated throughput, similarly the stable capacity region can outer bound the Shannon capacity region. The first reason is that the resource is not constantly contested under bursty conditions. The second reason is that the timing information rate contributes to the increase under bursty conditions.

Understanding the relationship between the information-theoretic capacity region and the stability region has received considerable attention in recent years and some progress has been made primarily for multiple access channels. Interestingly, the aforementioned regions (capacity and stability) are not in general identical and general conditions under which they coincide are known only in very few cases \cite{b:EphremidesHajekUnion}.

\section{Summary} \label{sec:conclusions}
In this work, we derived the stability region for the two-user broadcast channel. We considered two decoding schemes at the receiver side, namely treating interference as noise by both receivers and successive decoding by the strong receiver. For the latter, two simple power allocation policies were studied, a fixed power allocation and an adaptive power scheme based on the queues' states. Furthermore, we obtained conditions on the convexity for the stability region in order to identify regimes with superior performance comparing with time-sharing.

\bibliographystyle{IEEEtran}
\bibliography{thesis}

\end{document}